\documentclass[english,english,pra,english,preprint,amsmath,amssymb,aps,longbibliography,showkeys, titlepage]{revtex4-2}
\usepackage{lmodern}
\usepackage[T1]{fontenc}
\usepackage[latin9]{inputenc}
\usepackage{amsmath}
\usepackage{amsthm}
\usepackage{amssymb}
\usepackage{graphicx}

\makeatletter
\theoremstyle{plain}
\newtheorem{thm}{\protect\theoremname}
\theoremstyle{definition}
\newtheorem{defn}[thm]{\protect\definitionname}
\theoremstyle{plain}
\newtheorem{lem}[thm]{\protect\lemmaname}

\usepackage[T1]{fontenc}
\usepackage{amsmath}
\usepackage{amsthm}
\usepackage{amssymb}
\usepackage{graphicx}

\usepackage{graphicx}
\usepackage{float}
\usepackage{xcolor}
\usepackage[hypertexnames=false]{hyperref}
\hypersetup{
	breaklinks = true,
    colorlinks = true,
    citecolor = {blue},
	urlcolor = {blue},
	linkcolor = {blue}
}

\makeatother

\usepackage{babel}
\providecommand{\definitionname}{Definition}
\providecommand{\lemmaname}{Lemma}
\providecommand{\theoremname}{Theorem}

\begin{document}
\title{Topological Langmuir-cyclotron wave}
\author{Hong Qin}
\email{hongqin@princeton.edu }

\author{Yichen Fu}
\email{yichenf@princeton.edu}

\affiliation{Princeton Plasma Physics Laboratory and Department of Astrophysical
Sciences, Princeton University, Princeton, NJ 08540}
\begin{abstract}
A theoretical framework is developed to describe the Topological Langmuir-Cyclotron
Wave (TLCW), a recently identified topological surface excitation
in magnetized plasmas. As a topological wave, the TLCW propagates
unidirectionally without scattering in complex boundaries. The TLCW
is studied theoretically as a spectral flow of the Hamiltonian Pseudo-Differential-Operator
(PDO) $\hat{H}$ for waves in an inhomogeneous plasma. The semi-classical
parameter of the Weyl quantization for plasma waves is identified
to be the ratio between electron gyro-radius and the inhomogeneity
scale length of the system. Hermitian eigenmode bundles of the bulk
Hamiltonian symbol $H$ for plasma waves are formally defined. Because
momentum space in classical continuous media is contractible in general,
the topology of the eigenmode bundles over momentum space is trivial.
This is in stark contrast to condensed matters. Nontrivial topology
of the eigenmode bundles in classical continuous media only exists
over phase space. A boundary isomorphism theorem is established to
facilitate the calculation of Chern numbers of eigenmode bundles over
non-contractible manifolds in phase space. It also defines a topological
charge of an isolated Weyl point in phase space without adopting any
connection. Using these algebraic topological techniques and an index
theorem formulated by Faure, it is rigorously proven that the nontrivial
topology at the Weyl point of the Langmuir wave-cyclotron wave resonance
generates the TLCW as a spectral flow. It is shown that the TLCW can
be faithfully modeled by a tilted Dirac cone in phase space. An analytical
solution of the entire spectrum of the PDO of a generic tilted phase
space Dirac cone, including its spectral flow, is given. The spectral
flow index of a tilted Dirac cone is one, and its mode structure is
a shifted Gaussian function.
\end{abstract}
\maketitle

\section{Introduction \label{sec:Introduction}}

Topological wave in classical continuous media is an active research
topic for its practical importance. For example, it was discovered
\citep{delplace2017topological,Faure2019} that the well-known equatorial
Kelvin wave, which can trigger an El Nino episode \citep{Roundy2007},
is a topological wave in nature. Topological waves in cold magnetized
plasmas have been recently studied \citep{parker2020topological,fu2021topological}.
Through comprehensive numerical simulations and heuristic application
of the principle of bulk-edge correspondence, a topological surface
excitation called Topological Langmuir-Cyclotron Wave (TLCW) was identified
\citep{fu2021topological,Fu2022a}. As a topological wave, the TLCW
has topological robustness, i.e., it is unidirectional and free of
scattering and reflection. Thus, the TLCW excitation is expected to
be experimentally observable. 

In the present study, we develop a theoretical framework to describe
the TLCW and rigorously prove that it is produced by the nontrivial
topology at the Weyl point due to the Langmuir wave-cyclotron wave
resonance using an index theorem for spectral flows formulated by
Faure \citep{Faure2019} and tools of algebraic topology. Most of
the techniques developed are applicable to general topological waves
in classical continuous media as well. The key developments of the
present study are summarized as follows. 
\begin{enumerate}
\item The TLCW is theoretically described as a spectral flow of the global
Hamiltonian Pseudo-Differential-Operator (PDO) $\hat{H}$ for waves
in an inhomogeneous magnetized plasma. For this problem, the semi-classical
parameter of the Weyl quantization operator, which maps the bulk Hamiltonian
symbol $H$ to $\hat{H}$, is identified as the ratio between electron
gyro-radius and the scale length of the inhomogeneity. We emphasize
the important role of the semi-classical parameters and the necessity
to identify them for topological waves in classical continuous media
according to the nature of the physics under investigation. 
\item We formally construct the Hermitian eigenmode bundles of the bulk
Hamiltonian symbol $H$, and show that it is the topology of the eigenmode
bundles over non-contractible manifolds in phase space that determines
the properties of spectral flows for classical continuous media. We
show that the topology of eigenmode bundles on momentum (wavenumber)
space is trivial in classical continuous media, unlike in condensed
matters. Without modification, the Atiyah-Patodi-Singer (APS) index
theorem \citep{Atiyah1976} proved for spectral flows over $S^{1}$
is only applicable to condensed matters, and Faure's index theorem
\citep{Faure2019} for spectral flows over $\mathbb{R}$-valued wavenumbers
should be adopted for classical continuous media. 
\item A boundary isomorphism theorem (Theorem \ref{thm:BoundaryIso}) is
proved to facilitate the calculation of Chern numbers of eigenmode
bundles over a 2D sphere in phase space. The theorem also defines
a topological charge of an isolated Weyl point in phase space using
a topological method, i.e., without using any connection. 
\item An analytical solution of the global Hamiltonian PDO of a generic
tilted Dirac cone in phase space is found, which generalizes the previous
result for a straight Dirac cone \citep{Faure2019}. The spectral
flow index of a tilted phase space Dirac cone is calculated to be
one, and the mode structure of the spectral flow is found to be a
shifted Gaussian function.
\item These tools are applied to prove the existence of the TLCW in magnetized
plasmas with the spectral flow index being one. The Chern theorem
(Theorem \ref{thm:Chern}), instead of the Berry connection or any
other connection, was used to calculate the Chern numbers. And it
is shown that the TLCW can be faithfully described by a tilted Dirac
cone in phase space. 
\end{enumerate}
The paper is organized as follows. In Sec.\,\ref{sec:Problem-statement},
we pose the problem to be studied and describe the general properties
of the TLCW identified by numerical simulations. Section \ref{sec:Numerial}
presents additional numerical evidence and simulation results of the
TLCW. In Sec.\,\ref{sec:Topology}, we define the Hermitian eigenmode
bundles of waves in classical continuous media and develop algebraic
topological tools to study the nontrivial topology of eigenmode bundles
over phase space. The existence of TLCW as a spectral flow is proven
in Sec.\,\ref{sec:TLCWPrediction}. We construct a tilted Dirac cone
model for the TLCW in Sec.\,\ref{sec:AnalyticalTDC}, and the entire
spectrum of the PDO of a generic tilted Dirac cone, including its
spectral flow, is solved analytically.

\section{Problem statement and general properties of TLCW \label{sec:Problem-statement}}

We first pose the problem to be addressed in the present study, introduce
the governing equations, set up the class of equilibrium plasmas that
might admit the TLCW, and describe its general properties.

Consider a cold magnetized plasma with fixed ions. The equilibrium
magnetic field $\boldsymbol{B}_{0}=B_{0}\boldsymbol{e}_{z}$ is assumed
to be constant. Because the plasma is cold, any density profile $n(\boldsymbol{r})$
is an admissible equilibrium. Denote by $L\sim\left|n/\nabla n\right|$
the characteristic scale length of $n$. There is no equilibrium electrical
field and electron flow velocity, i.e., $\boldsymbol{v}_{0}=0$ and
$\boldsymbol{E}_{0}=0.$ The linear dynamics of the system is described
by the following equations for the perturbed electromagnetic field
$\boldsymbol{E}$ and $\boldsymbol{B}$, and the perturbed electron
flow $\boldsymbol{v}$. 
\begin{align}
 & \partial_{t}\boldsymbol{v}=-e\boldsymbol{E}/m_{e}-\Omega\boldsymbol{v}\times\boldsymbol{e}_{z},\label{eq:basic1}\\
 & \partial_{t}\boldsymbol{E}=c\nabla\times\boldsymbol{B}+4\pi en\boldsymbol{v},\label{eq:basic2}\\
\  & \partial_{t}\boldsymbol{B}=-c\nabla\times\boldsymbol{E},\label{eq:basic3}
\end{align}
where $\Omega=eB_{0}/m_{e}c$ is the cyclotron frequency, $m_{\mathrm{e}}$
is the electron mass, and $e>0$ is the elementary charge. We normalize
$\boldsymbol{v}$ by $1/\sqrt{4\pi n(\boldsymbol{r})m_{\mathrm{e}}}$,
$t$ by $1/\Omega$, $\boldsymbol{r}$ by $L$, and $\nabla$ by $1/L$.
In the normalized variables, Eqs.\,(\ref{eq:basic1})-(\ref{eq:basic3})
can be written as 
\begin{align}
\mathrm{i}\partial_{t}\psi & =\hat{H}\psi,\\
\psi & =\left(\begin{array}{c}
\boldsymbol{v}\\
\boldsymbol{E}\\
\boldsymbol{B}
\end{array}\right),\\
\hat{H} & (\boldsymbol{r},-\mathrm{i}\eta\nabla)=\begin{pmatrix}\mathrm{i}\boldsymbol{e}_{z}\times & -\mathrm{i}\omega_{\text{p}} & 0\\
\mathrm{i}\omega_{\text{p}} & 0 & \mathrm{i}\eta\nabla\times\\
0 & -\mathrm{i}\eta\nabla\times & 0
\end{pmatrix},\label{eq:Hhat}
\end{align}
where $\mathrm{i}\boldsymbol{e}_{z}\times$ and $\mathrm{i}\eta\nabla\times$
denote $3\times3$ anti-symmetric matrices corresponding to $\boldsymbol{e}_{z}$
and $\nabla$, respectively. For a generic vector $\boldsymbol{u}=(u_{x},u_{y},u_{z})$
in $\mathbb{R}^{3},$ the corresponding $3\times3$ anti-symmetric
matrix is defined to be 
\begin{equation}
\boldsymbol{u}\times\equiv\begin{pmatrix}0 & -u_{z} & u_{y}\\
u_{z} & 0 & -u_{x}\\
-u_{y} & u_{x} & 0
\end{pmatrix}.
\end{equation}
In Eq.\,(\ref{eq:Hhat}), $\omega_{\mathrm{p}}(\boldsymbol{r})=\sqrt{4\pi n(\boldsymbol{r})e^{2}/m_{\mathrm{e}}}/\Omega$
is the local plasma frequency normalized by $\Omega,$ and $\eta\equiv c/(L\Omega)\sim\rho_{e}/L$
is a dimensionless parameter proportional to the ratio between electron
gyro-radius and the scale length of $n$. Here, $\eta$ is assumed
to be small, i.e., $\eta\ll1$, and it plays the role of the semi-classical
parameter for the Weyl quantization of this problem.

The Weyl quantization operator 
\begin{equation}
\mathrm{Op}_{\eta}:f\rightarrow\hat{f}=\mathrm{Op}_{\eta}(f)
\end{equation}
maps a function in phase space $f(\boldsymbol{r},\boldsymbol{k})$,
called a symbol, to an Pseudo-Differential-Operator (PDO) $\hat{f}$
on functions $\psi(\boldsymbol{r})$ on the $n$-dimensional configuration
space. The operator $\hat{f}=\mathrm{Op}_{\eta}(f)$ is defined by
\begin{equation}
\hat{f}\psi(\boldsymbol{r})=\frac{1}{(2\pi\eta)^{n}}\int f\left(\frac{\boldsymbol{r}+\boldsymbol{s}}{2},\boldsymbol{k}\right)\exp\left(\frac{\mathrm{i}\boldsymbol{k}\cdot\left(\boldsymbol{x}-\boldsymbol{y}\right)}{\eta}\right)\psi(\boldsymbol{s})\mathrm{d}\boldsymbol{s}\mathrm{d}\boldsymbol{k}\thinspace.
\end{equation}
In particular, we have $\hat{\boldsymbol{k}}=-\mathrm{i}\eta\nabla$.

For the $\hat{H}$ given by Eq.\,(\ref{eq:Hhat}), its pre-image
$H$, i.e., the symbol $H$ satisfying $\hat{H}=\mathrm{Op}_{\eta}(H)$,
is 
\begin{equation}
H(\boldsymbol{r},\boldsymbol{k})=\begin{pmatrix}\mathrm{i}\boldsymbol{e}_{z}\times & -\mathrm{i}\omega_{\text{p}} & 0\\
\mathrm{i}\omega_{\text{p}} & 0 & -\boldsymbol{k}\times\\
0 & \boldsymbol{k}\times & 0
\end{pmatrix}.\label{eq:H}
\end{equation}

In quantum theory, the semi-classical parameter is typically the Plank
constant $\hbar$, and it is a crucial parameter in the index theorems
for spectral flow \citep{Atiyah1976,Faure2019} of PDOs. For the plasma
waves in the present study, the semi-classical parameter is identified
to be $\eta\equiv c/(L\Omega)$, which is the ratio between electron
gyro-radius and the scale length of the equilibrium plasma. Notice
that in the PDO $\hat{H}(\boldsymbol{r},-\mathrm{i}\eta\nabla)$ the
differential operator $\nabla$ is normalized by $1/L$, but in the
symbol $H(\boldsymbol{r},\boldsymbol{k})$ the wavenumber $\boldsymbol{k}$
is normalized by $\Omega/c$, thanks to the small semi-classical parameter
$\eta$ strategically placed in the Weyl quantization operator $\mathrm{Op}_{\eta}$.
This structure between the PDO and the symbol is required for the
application of the index theorem of spectral flow \citep{Atiyah1976,Faure2019}.
In the study of other topological properties of classical media, such
as electromagnetic materials \citep{silveirinha2015chern,silveirinha2016bulk,gangaraj2017berry,marciani2020chiral},
fluid systems \citep{delplace2017topological,Faure2019,perrot2019topological,tauber2019bulk,venaille2021wave,zhu2021topology,souslov2019topological,qin2019kelvin,fu2020physics,David2022},
and magnetized plasmas \citep{gao2016photonic,yang2016one,parker2020nontrivial,parker2020topological,parker2021topological,fu2021topological,Fu2022a,Rajawat2022,qin2021spontaneous},
we believe that appropriate semi-classical parameters should also
be carefully determined first based on the specific nature of the
problems under investigation.

In plasma physics, the symbol $H(\boldsymbol{r},\boldsymbol{k})$
is called the local Hamiltonian of the system, but it is known as
the bulk Hamiltonian in condensed matter physics. Thus, in the present
context, the phrases ``bulk modes'' and ``local modes'' have the
same meaning, referring to the spectrum determined by $H(\boldsymbol{r},\boldsymbol{k})$
locally at each $\boldsymbol{r}$ and each $\boldsymbol{k}$ separately.
The spectrum of the PDO $\hat{H}(\boldsymbol{r},-\mathrm{i}\eta\nabla)$
will be called global modes. The edge modes, including topological
edge modes, refer to the global modes of $\hat{H}(\boldsymbol{r},-\mathrm{i}\eta\nabla)$
whose mode structures are non-vanishing only in some narrow interface
regions. It is unfortunate that the phrases ``local modes'' and
``edge modes'', defined in different branches of physics, have very
different meanings.

For a fixed $\boldsymbol{r}$ and a fixed $\boldsymbol{k}$, $H(\boldsymbol{r},\boldsymbol{k})$
is a $9\times9$ Hermitian matrix. Denote its 9 eigenmodes by 
\[
(\omega_{n},\psi_{n}),\thinspace\thinspace n=-4,-3,\cdots,3,4\thinspace,
\]
which are ordered by the value of the eigenfrequencies, i.e., $\omega_{i}\leq\omega_{j}$
for $i<j$. Under this index convention, it can be verified that $\omega_{-n}=-\omega_{n}$
and $\omega_{0}=0$, i.e., the spectrum is symmetric with respect
to the real axis. Plotted in Fig.\,\ref{fig:dispersion_relation}
are the dispersion relations of $\omega_{n}$ $(n=1,2,3,4)$ for an
over-dense and an under-dense plasma, respectively. The eigenfrequencies
are plotted as functions of $k_{z}$ and $k_{y}$ only since the spectrum
is invariant when $\boldsymbol{k}$ rotates in the $x$-$y$ plane.

\begin{figure}[ht]
\centering \includegraphics[width=11cm]{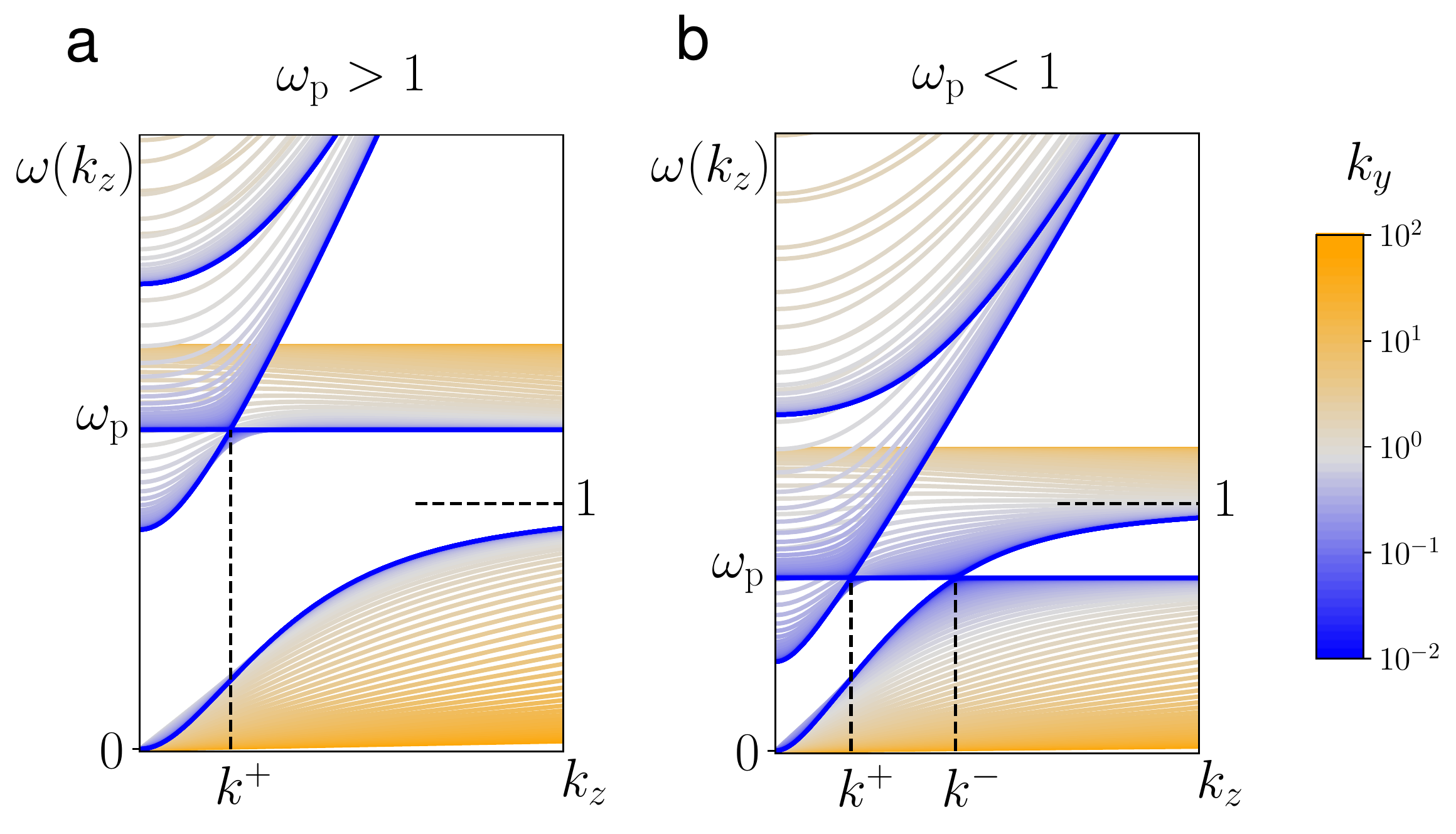} \caption{The dispersion relation $\omega_{n}(k_{z},k_{y})$ $(n=1,2,3,4)$
for (a) an over-dense plasma and (b) an under-dense plasma. Different
values of $k_{y}$ are indicated by the color map.}
\label{fig:dispersion_relation}
\end{figure}

Straightforward analysis shows that for a given $\omega_{\text{p}},$
the spectrum has two possible resonances, a.k.a. Weyl points, when
$\boldsymbol{k}_{\perp}=0$ and $k_{z}=k^{\pm}$, where 
\begin{equation}
k^{\pm}\equiv\dfrac{\omega_{\text{p}}}{\sqrt{1\pm\omega_{\text{p}}}}
\end{equation}
are two critical wavenumbers for the given $\omega_{\text{p}}.$ We
are interested in the resonance at $\boldsymbol{k}_{\perp}=0$ and
$k_{z}=k^{-}$, which is between the Langmuir wave and the cyclotron
wave (the R-wave near the cyclotron frequency). Obviously, this Langmuir-Cyclotron
(LC) resonance or Weyl point exists when and only when the plasma
is under-dense, i.e., $\omega_{\text{p}}<1$.

For a given $k_{z}$, the LC resonance occurs when $\omega_{\text{p}}=\omega_{\text{pc}}$,
where 
\begin{equation}
\omega_{\text{pc}}\equiv\dfrac{\sqrt{k_{z}^{4}+4k_{z}^{2}}-k_{z}^{2}}{2}
\end{equation}
is the critical plasma frequency for the given $k_{z}$. In the parameter
space of $(\omega_{\text{p}},k_{x},k_{y})$ for a fixed $k_{z}$,
when moving away from the LC Weyl point $(\omega_{\text{p}},k_{x},k_{y})=(\omega_{\text{pc}},0,0)$,
the distance between $\omega_{1}$ and $\omega_{2}$ will increase,
i.e., 
\begin{equation}
\omega_{2}-\omega_{1}>0,\,\,\text{when }(\omega_{\text{p}},k_{x},k_{y})\ne(\omega_{\text{pc}},0,0).
\end{equation}
Interesting topological physics happens in the neighborhood of the
LC Weyl point $(\omega_{\text{p}},k_{x},k_{y})=(\omega_{\text{pc}},0,0)$.
Figure \ref{fig:DiracCone} shows the surfaces of $\omega_{1}$ and
$\omega_{2}$ as functions of $\omega_{\text{p}}$ and $k_{x}$ near
the LC Weyl point. The structure is known as a Dirac cone. One important
feature of the Dirac cone at the LC Weyl point is that it is tilted.
Also, note that this tilted Dirac cone is in phase space since $\omega_{\text{p}}$
is a function of $x.$ This is different from condensed matter physics,
where the Dirac cone is mostly in momentum space.

\begin{figure}[ht]
\includegraphics[width=8cm]{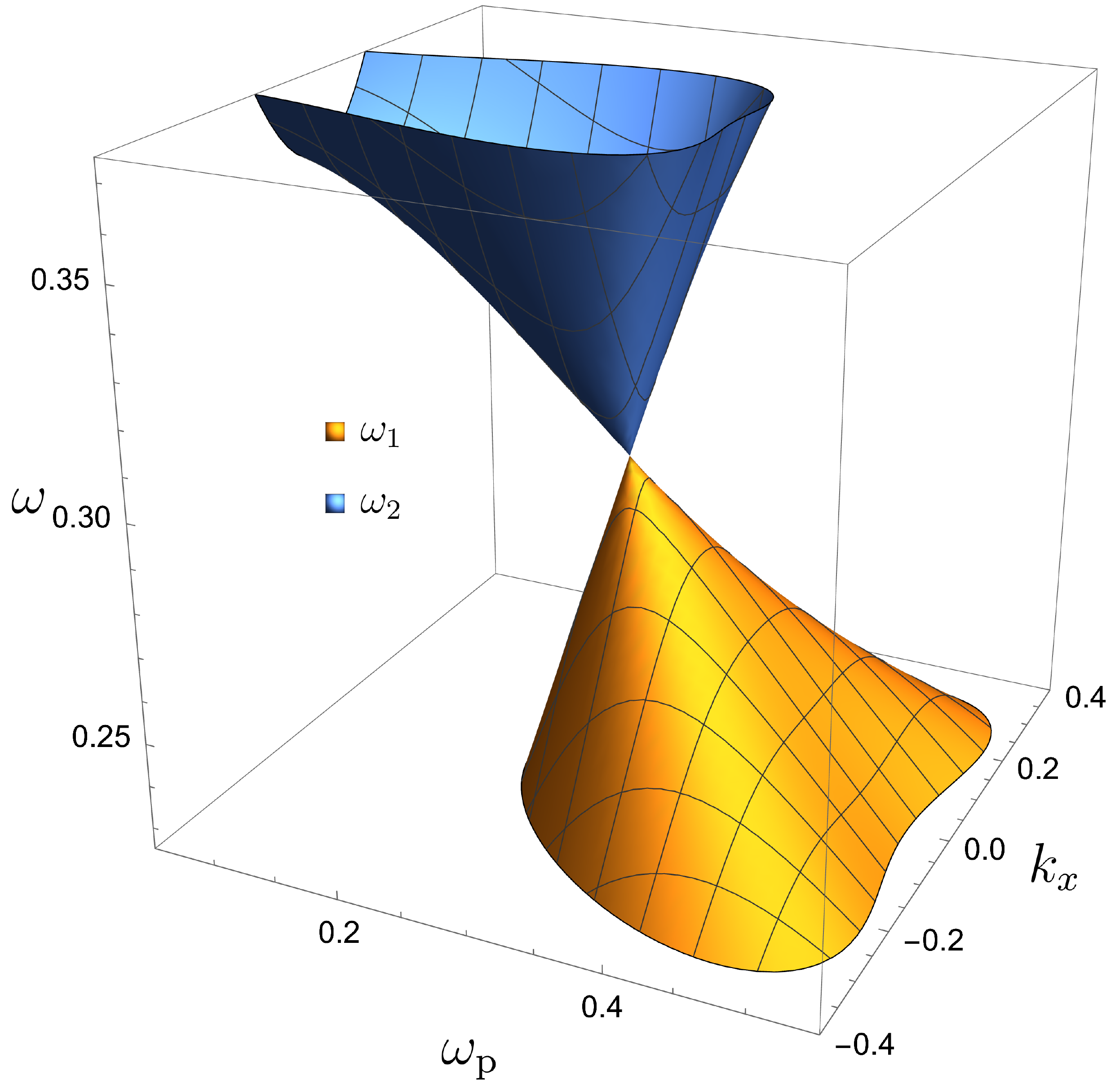} \caption{Tilted phase space Dirac cone in the neighborhood of the LC Weyl point
$(\omega_{\text{p}},k_{x},k_{y})=(\omega_{\text{pc}},0,0)$. }
\label{fig:DiracCone}
\end{figure}

Previous numerical studies and qualitative consideration \citep{fu2021topological,Fu2022a}
indicated that for a given $k_{z},$ a simple 1D equilibrium that
is inhomogeneous in the $x$-direction will admit the TLCW if the
range of $\omega_{\text{p}}(x)$ includes $\omega_{\text{pc}}.$ In
particular, we will consider the equilibrium profile displayed in
Fig.\,\ref{fig:1DGeometry}. The profile is homogeneous in Region
I $(x\le-1)$ and Region II $(x\ge1)$, and $\omega_{\text{p}}(x)$
monotonically decrease in the transition region $(-1\le x\le1)$.
The profile of $\omega_{\text{p}}(x)$ satisfies the condition 
\begin{gather}
\omega_{\text{p}1}>\omega_{\text{p}}(0)=\omega_{\text{pc}}>\omega_{\text{p}2}\thinspace,\label{eq: condition1}\\
\omega_{\text{p}1}\equiv\omega_{\text{p}}(x\le-1)\,,\\
\omega_{\text{p}2}\equiv\omega_{\text{p}}(x\ge1).
\end{gather}
The LC Weyl point locates at $x=0,$ and $\omega_{\text{p}1}$ is
the plasma frequency of Region I and $\omega_{\text{p}2}$ that of
Region II. Note that here $x$ is the dimensionless length normalized
by $L$, the scale length of equilibrium density profile $\omega_{\text{p}}(x)$.

\begin{figure}[ht]
\centering \includegraphics[width=8cm]{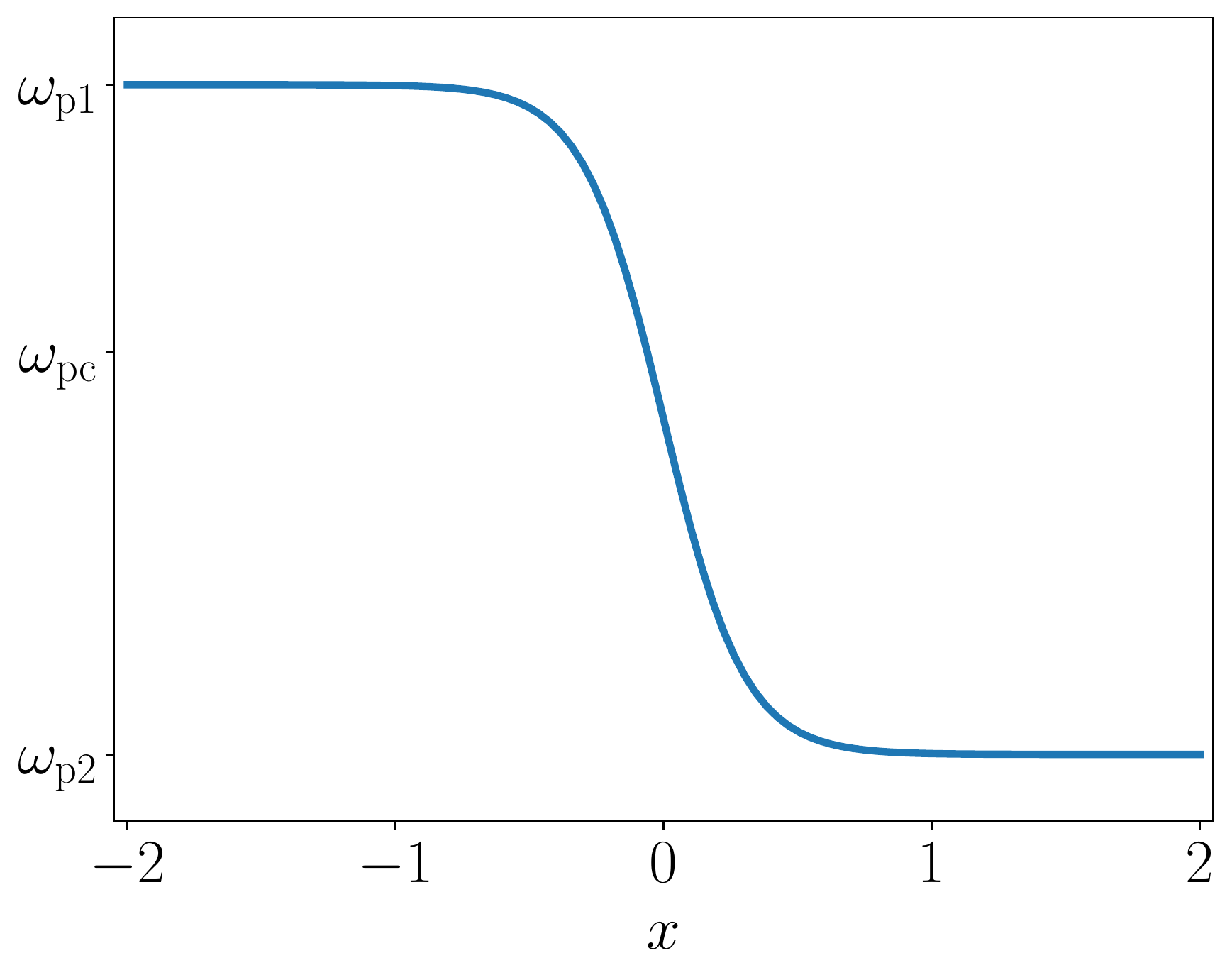} \caption{One-dimensional equilibrium with one transition region. The $x$ coordinate
has been normalized by $L,$ the scale length of $n(x)$.}
\label{fig:1DGeometry}
\end{figure}

In the following analysis, we will assume $k_{z}$ is a fixed parameter
unless explicitly stated otherwise.

Through the variation of $\omega_{\text{p}}(x)$, the spectrum $\omega_{n}(x,k_{x},k_{y})$
of the bulk Hamiltonian symbol $H(\boldsymbol{r},\boldsymbol{k})$
becomes a function of $x$. We defined the common gap condition for
the spectra $\omega_{1}(x,k_{x},k_{y})$ and $\omega_{2}(x,k_{x},k_{y})$
as follows. 
\begin{defn}
\label{def:CommGap} The spectra $\omega_{1}(x,k_{x},k_{y})$ and
$\omega_{2}(x,k_{x},k_{y})$ are said satisfying the common gap condition
for parameters exterior to the ball $B_{r}^{3}\equiv\left\{ (x,k_{x},k_{y})\mid x^{2}+k_{x}^{2}+k_{y}^{2}\le r^{2}\right\} $
in the phase space of $(x,k_{x},k_{y})$, if there exists an interval
$[g_{1}(r),g_{2}(r)]$ such that $\omega_{1}(x,k_{x},k_{y})<g_{1}(r)$
and $\omega_{2}(x,k_{x},k_{y})>g_{2}(r)$ for all $(x,k_{x},k_{y})\notin B_{r}^{3}.$
We call $[g_{1}(r),g_{2}(r)]$ the common gap of $\omega_{1}(x,k_{x},k_{y})$
and $\omega_{2}(x,k_{x},k_{y})$ for parameters exterior to the ball
$B_{r}^{3}$.

For all the parameter space that we have explored, condition (\ref{eq: condition1})
implies the common gap condition of $\omega_{1}(x,k_{x},k_{y})$ and
$\omega_{2}(x,k_{x},k_{y})$ for $(x,k_{x},k_{y})$ exterior to the
ball of $B_{1}^{3}$. Due to the algebraic complexity of $H(\boldsymbol{r},\boldsymbol{k})$,
this fact cannot be proved through a simple procedure, even though
no contour example was found numerically. In Sec.\,\ref{sec:AnalyticalTDC},
we will give a proof of this fact for a reduced Hamiltonian corresponding
to a tilted Dirac cone in the neighborhood of the LC Weyl point. In
the analysis before Sec.\,\ref{sec:AnalyticalTDC}, we will take
the common gap condition as an assumption.

For the 1D equilibrium with inhomogeneity in the $x$-direction, $k_{y}$
and $k_{z}$ are good quantum numbers and can be treated as system
parameters. The PDO $\hat{H}(\boldsymbol{r},-\mathrm{i}\eta\nabla)$
defined in Eq.\,(\ref{eq:Hhat}) reduces to 
\begin{equation}
\hat{H}(x,-\mathrm{i}\eta\partial_{x},k_{y},k_{z})=\begin{pmatrix}\mathrm{i}\boldsymbol{e}_{z}\times & -\mathrm{i}\omega_{\text{p}}(x) & 0\\
\mathrm{i}\omega_{\text{p}}(x) & 0 & (\mathrm{i}\eta\partial_{x},-k_{y,}-k_{z})\times\\
0 & (-\mathrm{i}\eta\partial_{x},k_{y,}k_{z})\times & 0
\end{pmatrix},\label{eq:Hhatx}
\end{equation}
and the corresponding bulk Hamiltonian symbol is 
\begin{equation}
H(x,k_{x},k_{y},k_{z})=\begin{pmatrix}\mathrm{i}\boldsymbol{e}_{z}\times & -\mathrm{i}\omega_{\text{p}}(x) & 0\\
\mathrm{i}\omega_{\text{p}}(x) & 0 & (-k_{x},-k_{y,}-k_{z})\times\\
0 & (k_{x},k_{y,}k_{z})\times & 0
\end{pmatrix},\label{eq:Hx}
\end{equation}
In Region I or II, the system is homogeneous, and in each region separately
it is valid to speak of the homogeneous eigenmodes of $\hat{H}(x,-\mathrm{i}\eta\partial_{x},k_{y},k_{z})$,
which are identical to the bulk modes of $H(x,k_{x},k_{y},k_{z})$
in that region.

The TLCW is a global eigenmode of $\hat{H}(x,-\mathrm{i}\eta\partial_{x},k_{y},k_{z})$
localized in the transition region of $-1<x<1$. Hence the name of
edge mode. In Fig.\,\ref{fig:1Dspectrum}, the numerically calculated
spectrum of $\hat{H}(x,-\mathrm{i}\eta\partial_{x},k_{y},k_{z})$
is plotted as a function of $k_{y}$. The spectrum consists of three
parts. The upper and lower parts are the spectrum of $\hat{H}(x,-\mathrm{i}\eta\partial_{x},k_{y},k_{z})$
that fall in the bulk bands of $H(x,k_{x},k_{y},k_{z})$ in Regions
I and II. The spectrum in the middle is a single line trespassing
the common band gap shared by Regions I and II. It is the TLCW. Its
frequency increases monotonically with $k_{y}$, passing through $\omega_{\text{pc}}$.
Such a curve of the dispersion relation for the edge mode as a function
of $k_{y}$ is known as a spectral flow because it ships one eigenmode
of $H(x,k_{x},k_{y},k_{z})$ from the lower band to the upper band
across the band gap (see Fig.\,\ref{fig:1Dspectrum}a). If there
were two or more edge modes in the gap as in the case of oceanic equatorial
waves \citep{delplace2017topological,Faure2019}, there would be two
or more spectral flows. In Sec.\,\ref{sec:TLCWPrediction}, we will
formally define spectral flow and show that the number of spectral
flow reflects the topology of the plasma waves and is determined by
a topological index known as the Chern number of a properly chosen
manifold in the parameter space. This is why they are called topological
edge modes. For the TLCW, we will show that its Chern number is one. 
\end{defn}

\begin{figure}[ht]
\centering \includegraphics[height=8cm]{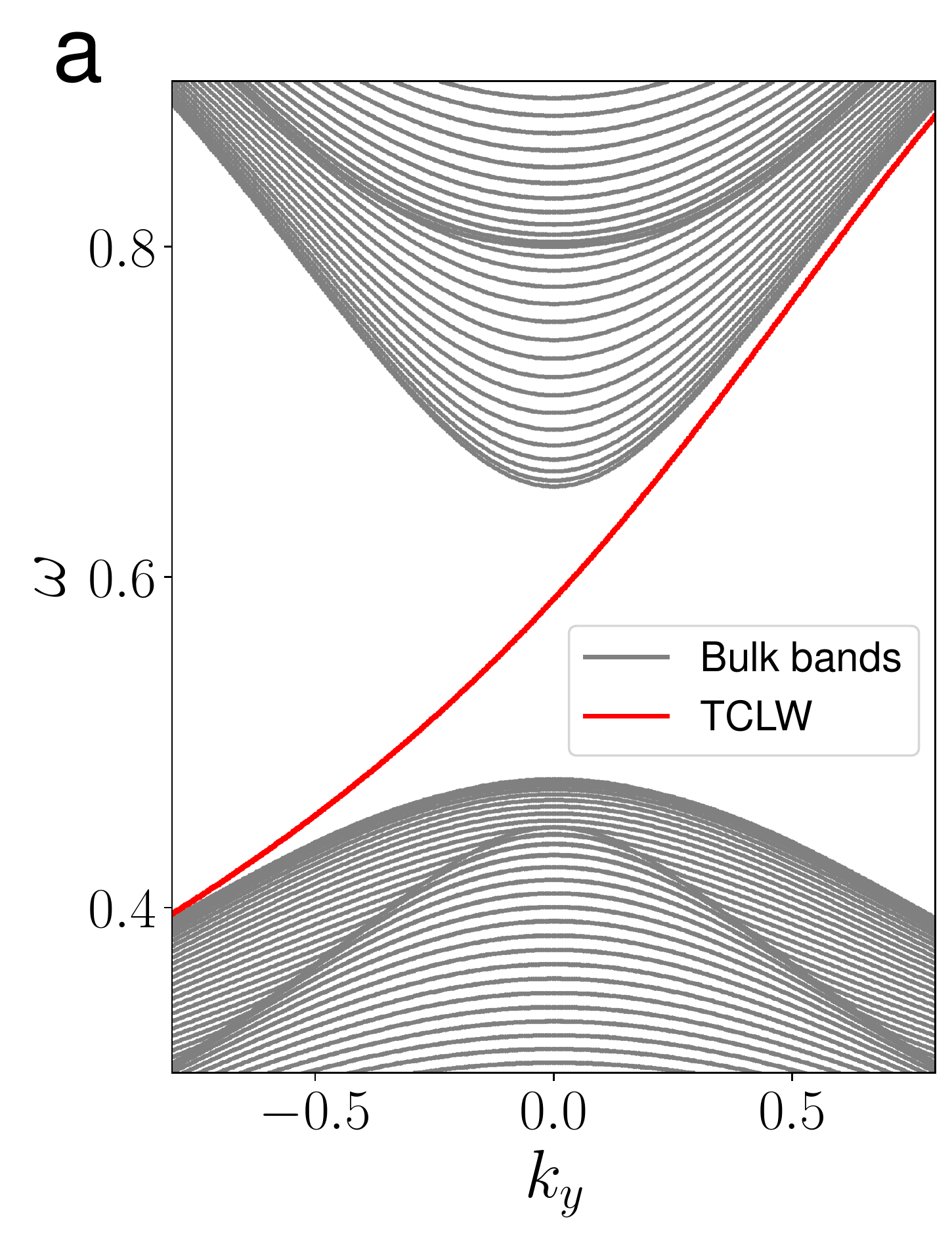} \hspace{1cm} \includegraphics[height=8cm]{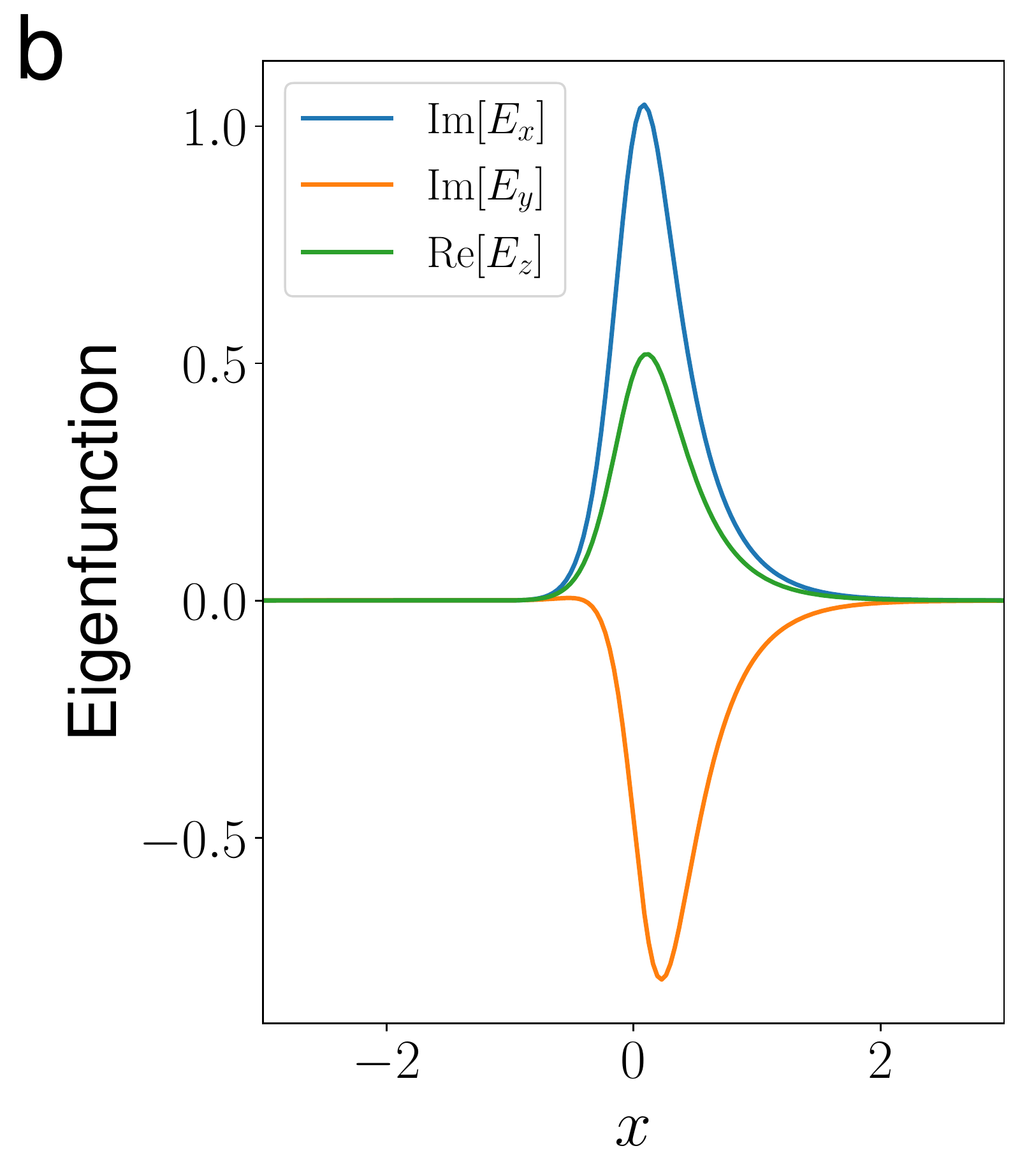}
\caption{(a) Spectrum of $\hat{H}(x,-\mathrm{i}\eta\partial_{x},k_{y},k_{z})$
as a function of $k_{y}$. (b) The mode structure of the TLCW.}
\label{fig:1Dspectrum}
\end{figure}

Condition (\ref{eq: condition1}) was identified as that for the existence
of the TLCW \citep{fu2021topological,Fu2022a} by heuristically applying
the bulk-edge correspondence using the numerically integrated values
of the Berry curvature over the $k_{x}$-$k_{y}$ plane for the bulk
modes of $H(x,k_{x},k_{y},k_{z})$ in Regions I and II. However, such
an integral should not be used as a topological index for the topology
waves in classical continuous media, because the topology of vector
bundles over a contractible base manifold is trivial, and the $k_{x}$-$k_{y}$
plane is contractible. A detailed discussion about the trivial topology
over the momentum space for classical continuous media can be found
in Sec.\,\ref{subsec:TrivialTop}.

In Secs.\,\ref{sec:Topology}-\ref{sec:AnalyticalTDC}, we show how
to formulate the bulk-edge correspondence for this problem in the
classical continuous media, using an index theorem of spectral flow
over wavenumbers taking values in $\mathbb{R}$ established by Faure
\citep{Faure2019} and techniques of algebraic topology. We rigorously
prove that there exists one TLCW when condition (\ref{eq: condition1})
and the common band gap condition are satisfied. After presenting
additional numerical evidence of the TLCW in the next section, we
will start our analytical study in Sec.\,\ref{sec:Topology} by defining
the Hermitian eigenmode bundle of plasma waves, with which the index
theorem is concerned with.

\section{Additional numerical evidence of TLCW \label{sec:Numerial}}

In this section, we display several more examples of numerically calculated
TLCW by a 1D eigenmode solver of $\hat{H}(x,-\mathrm{i}\eta\partial_{x},k_{y},k_{z})$
\citep{fu2021topological} as well as 3D time-dependent simulations
\citep{Fu2022a}.

The first example is the TLCW in a 1D equilibrium with two LC Wely
points, as illustrated in Fig.\,\ref{fig:1DGeometry2Edge}. The high-density
region is in the middle and the low-density region is on the two sides.
When condition (\ref{eq: condition1}) is satisfied, we expect to
observe two TLCWs, one on the right LC Weyl point and one on the left.
The numerically solved spectrum of $\hat{H}(x,-\mathrm{i}\eta\partial_{x},k_{y},k_{z})$
is shown in Fig.\,\ref{fig:1Dspectrum2Edege}, which meets the expectation
satisfactorily. 

\begin{figure}[ht]
\centering \includegraphics[width=9cm]{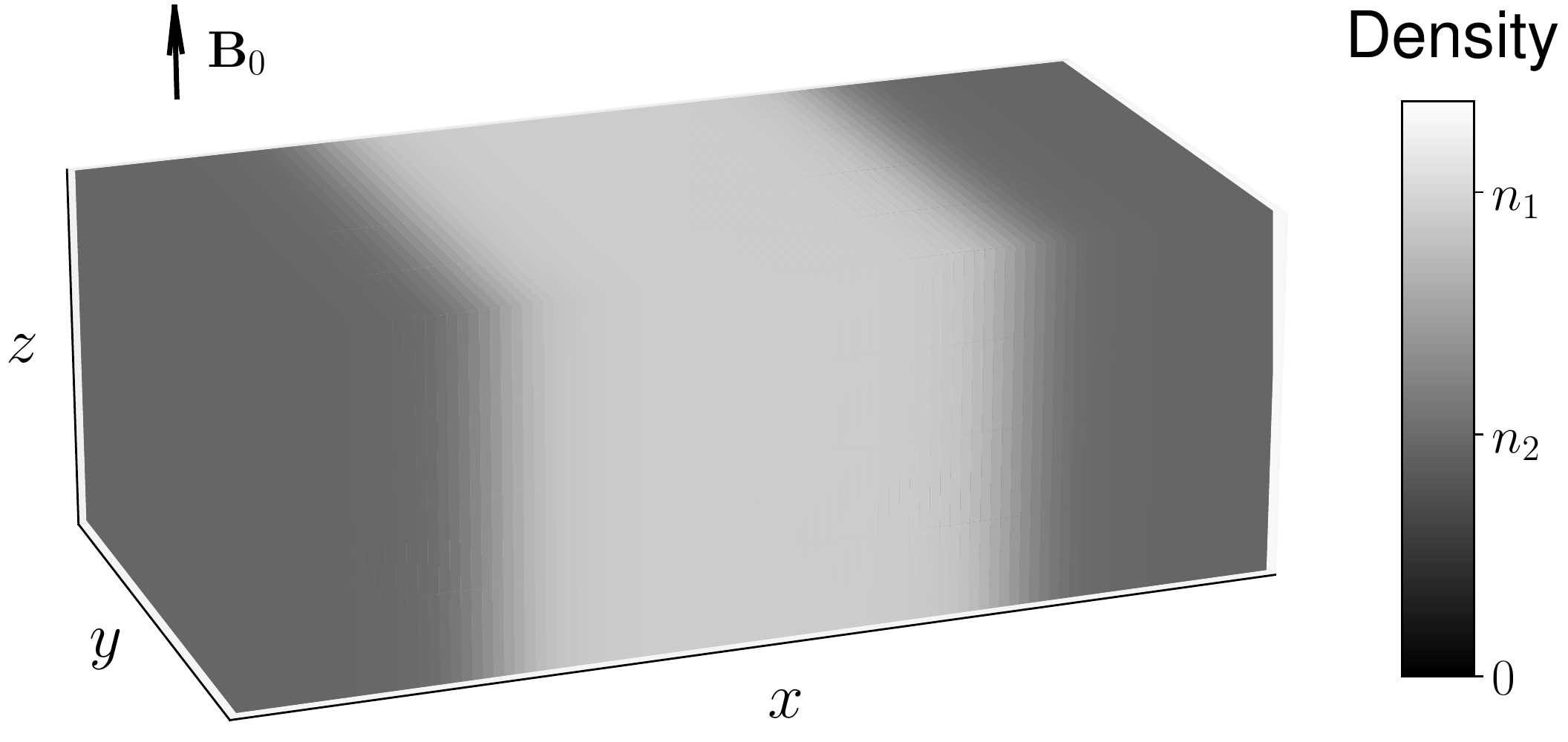} \caption{One-dimensional equilibrium with two Weyl points, located at the two
regions where density changes.}
\label{fig:1DGeometry2Edge}
\end{figure}

\begin{figure}[ht]
\centering \includegraphics[height=8cm]{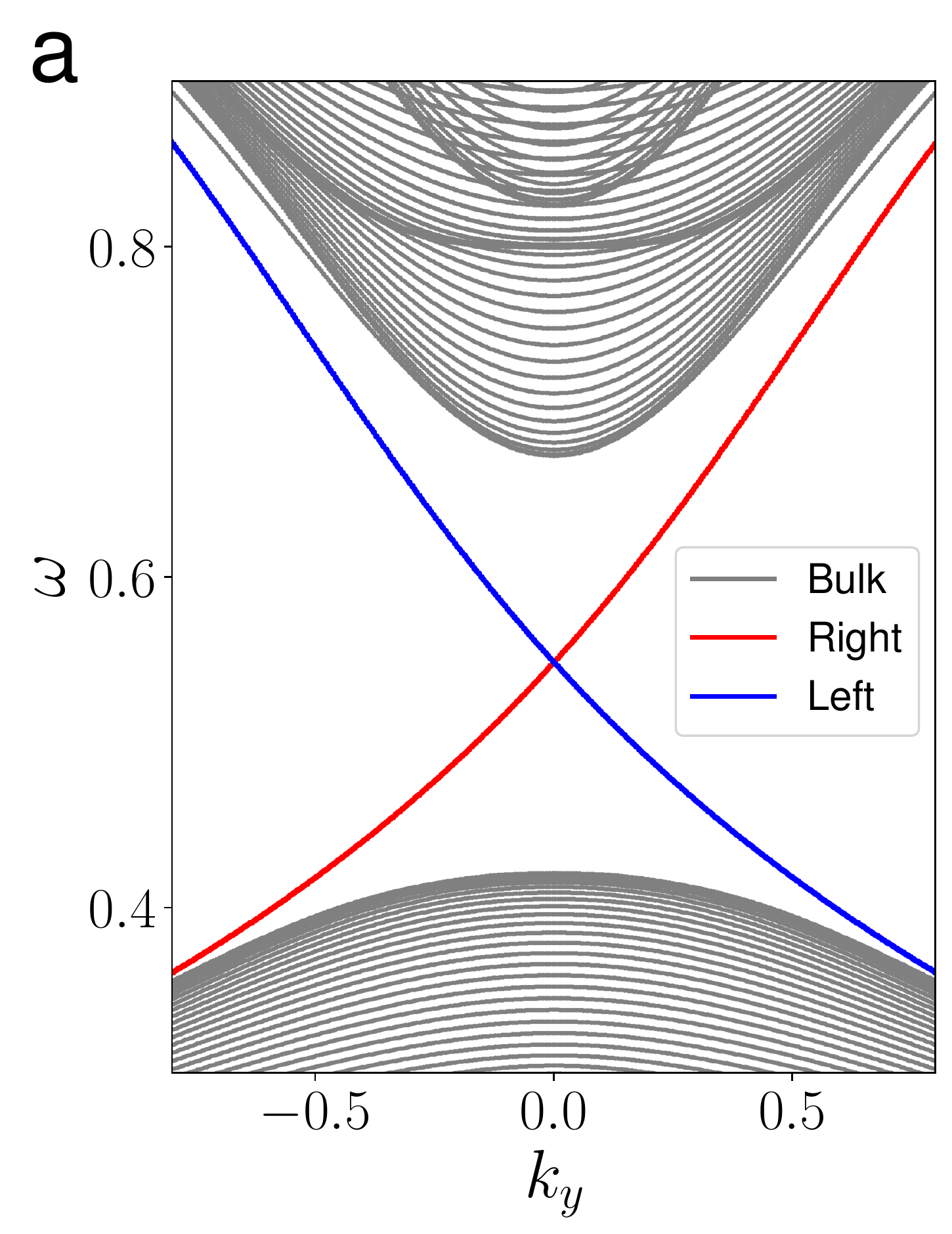} \hspace{0.5cm}
\includegraphics[height=8cm]{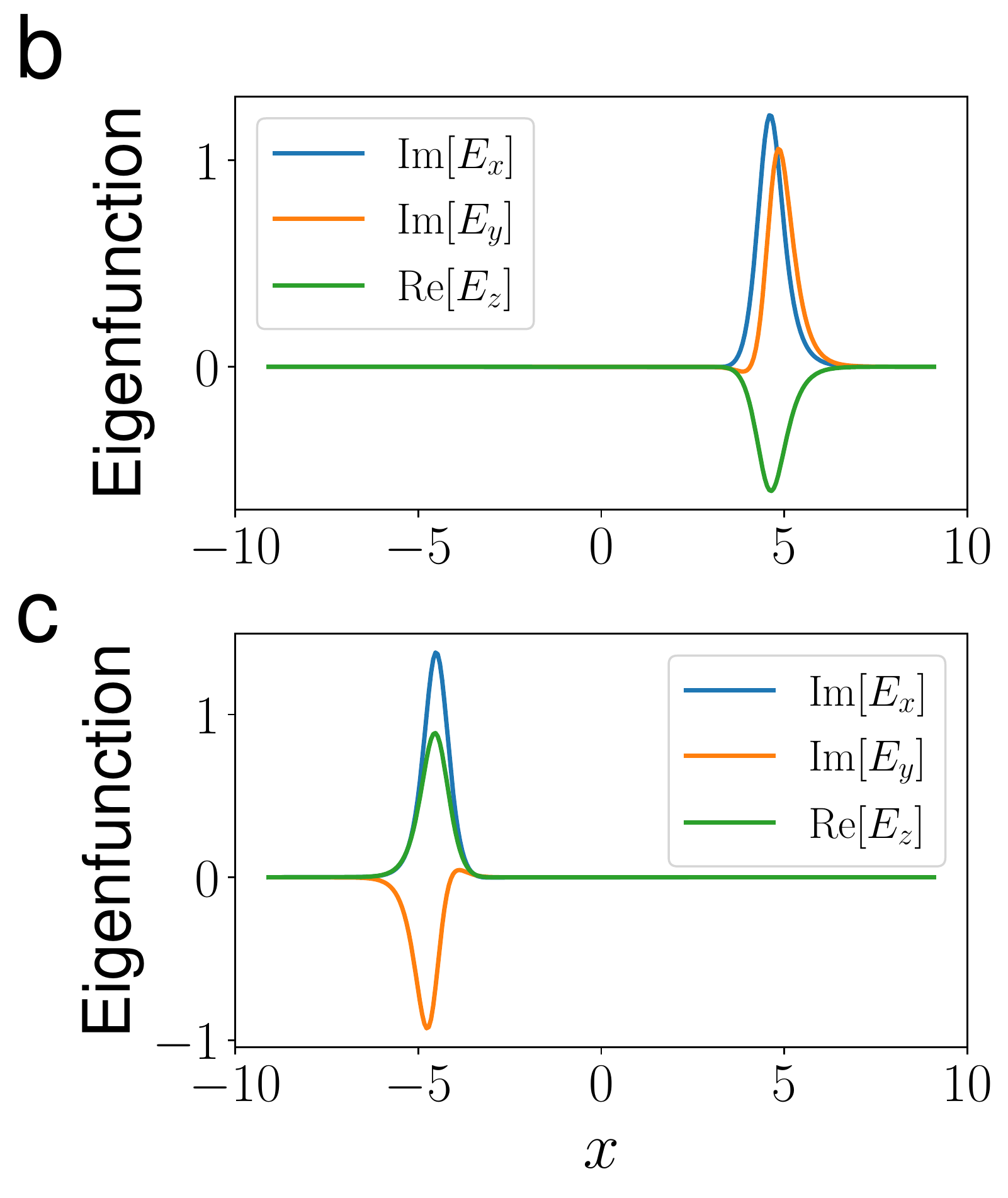} \caption{(a) Spectrum of $\hat{H}(x,-\mathrm{i}\eta\partial_{x},k_{y},k_{z})$
as a function of $k_{y}$. (b) The mode structure of the right TLCW.
(c) The mode structure of the left TLCW.}
\label{fig:1Dspectrum2Edege}
\end{figure}

Shown in Figs.\,\ref{fig:zigzag} and \ref{fig:oval} are 3D simulations
of the TLCW, where the boundary between two regions are nontrivial
curves in a 2D plane. In the simulations, an electromagnetic source
is placed on the boundary marked by the yellow star. For the simulation
in Fig.\,\ref{fig:zigzag}, the boundary is an irregular zigzag line.
As anticipated, the TLCW propagates along the irregular boundary unidirectionally
and without any scattering and reflection by the sharp turns. In Fig.\,\ref{fig:oval},
the boundary is a closed oval, and the TLCW stays on the oval boundary
as expected. The propagation is again unidirectionally and without
any scattering into other modes. Because $\omega_{\mathrm{p,1}}>\omega_{\mathrm{p,2}}$,
the TLCW propagates counterclockwise and carries a non-zero (kinetic)
angular momentum \citep{Fu2022a}. Even though the source does not
carry any angular momentum, an angular-momentum-carrying surface wave
is generated by the mechanism of the TLCW. 

\begin{figure}[ht]
\centering

\includegraphics[width=7cm]{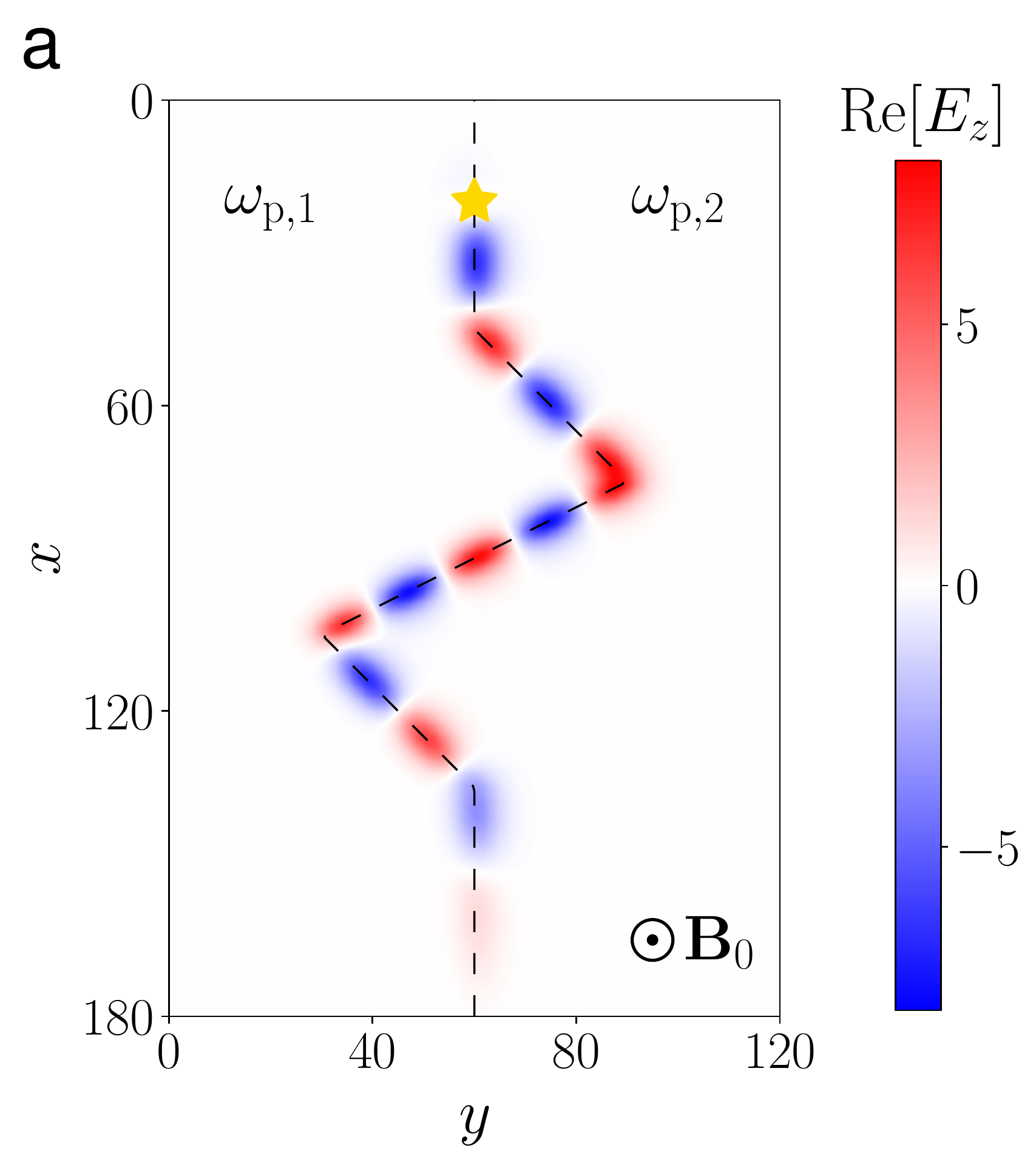}\includegraphics[width=7cm]{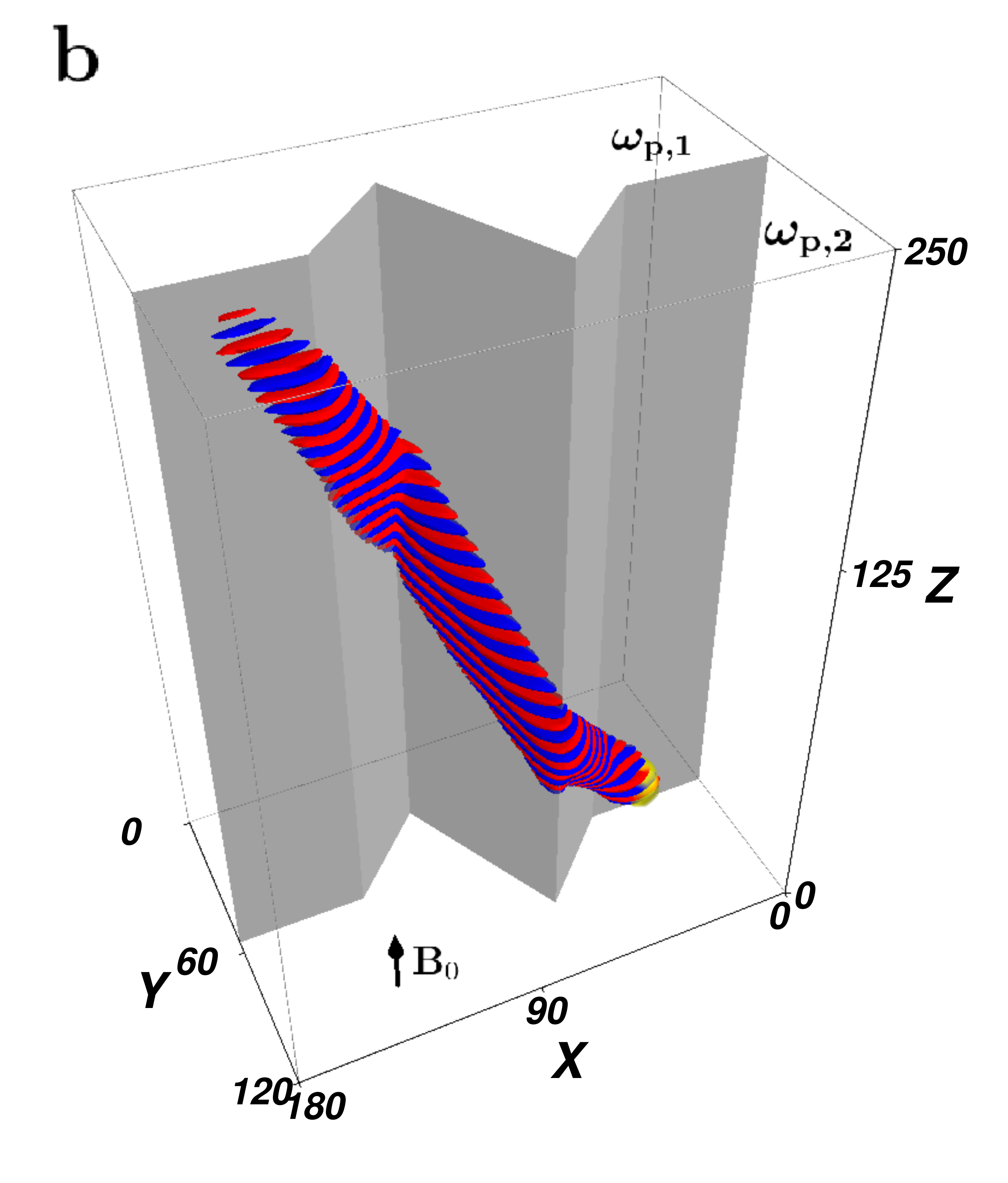}
\caption{(a) 2D and (b) 3D simulations of the TLCW excited on a zig-zag boundary.}
\label{fig:zigzag}
\end{figure}

\begin{figure}[ht]
\centering \includegraphics[height=8cm]{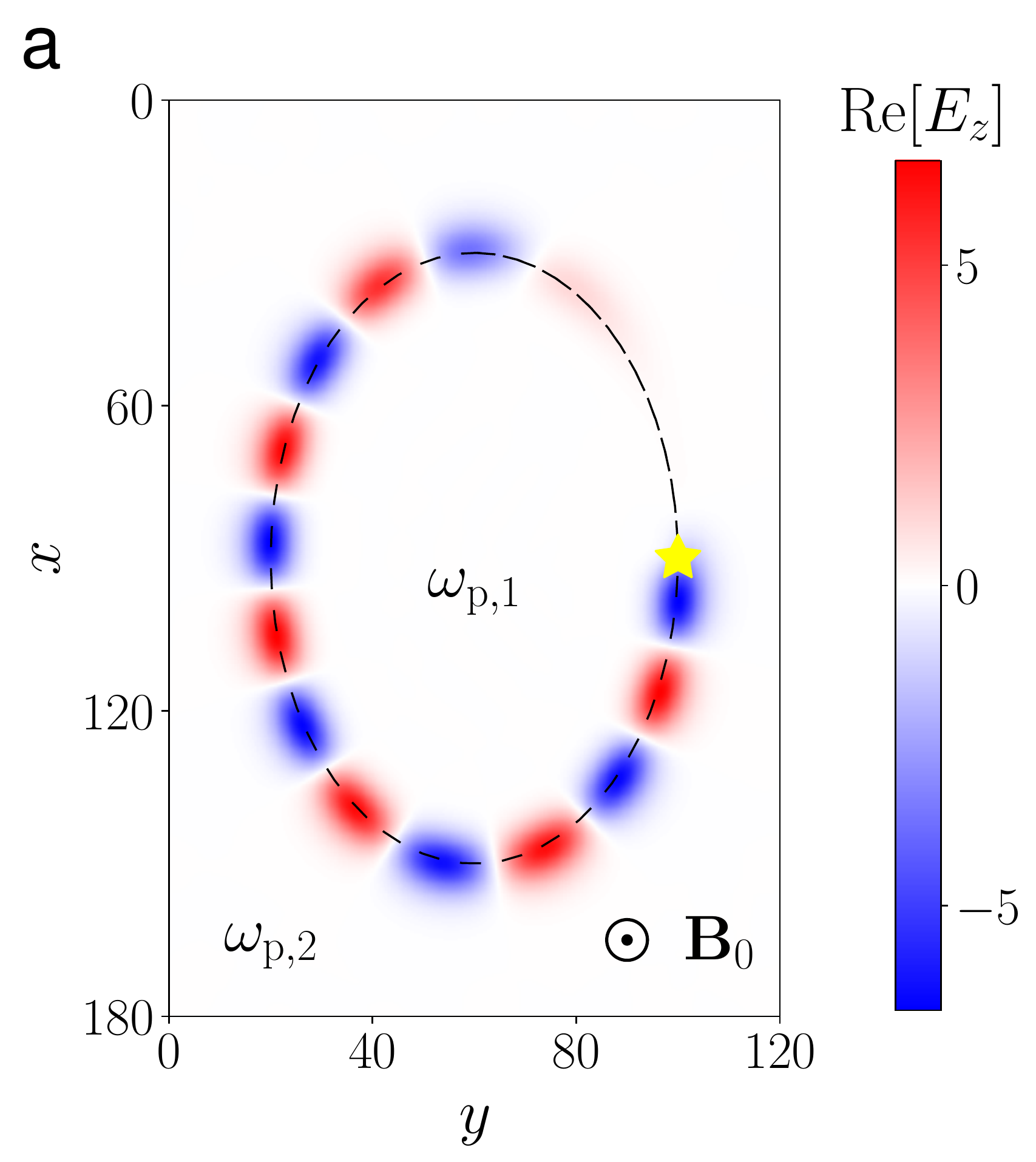} \hspace{1cm} \includegraphics[height=8cm]{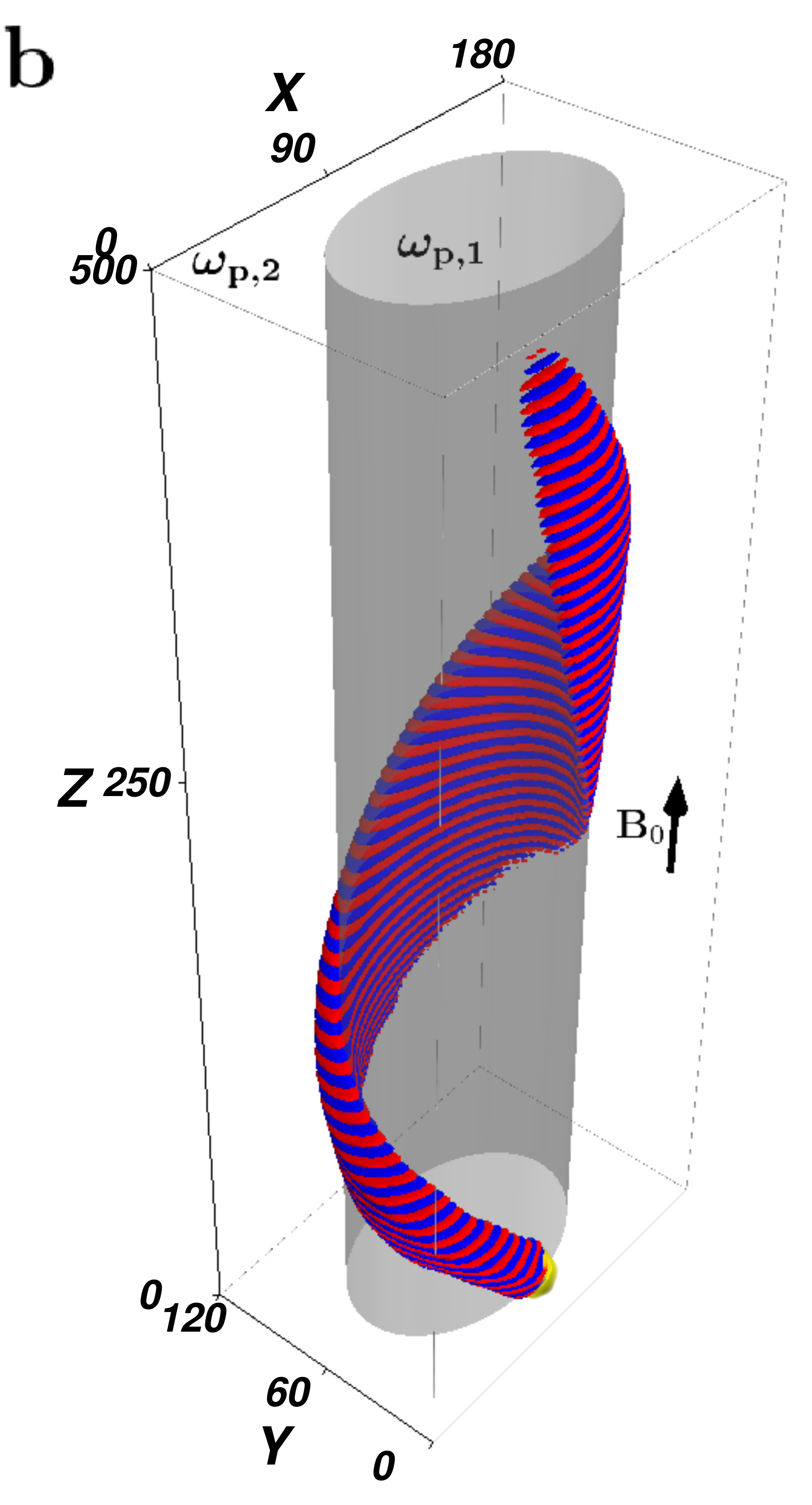}
\caption{(a) 2D and (b) 3D simulations of the TLCW excited on an oval boundary.}
\label{fig:oval}
\end{figure}

\section{Topology of Hermitian eigenmode bundles of plasma waves and waves
in classical continuous media\label{sec:Topology}}

The index theorem \citep{Atiyah1976,Faure2019} establishes the bulk-edge
correspondence linking the nontrivial topology of the bulk mode of
symbol $H$ and the spectral flow of PDO $\hat{H}$. The topology
here refers to that of the Hermitian bundles of eigenmodes over appropriate
regions of the parameter spaces, which we now define.

Denote the parameter space by $M$. In the present context, $M$ is
the space of all possible 4-tuples $(x,k_{x},k_{z},k_{z})$. For a
given $m=(x,k_{x},k_{z},k_{z})\in M$, the bulk Hamiltonian symbol
$H(m)$ supports a finite number of eigenmodes. For the plasma wave
operator defined by Eq.\,(\ref{eq:Hx}), there are 9 eigenmodes as
explained in Sec.\,\ref{sec:Problem-statement}. But, most of the
discussion and results in this section are not specific to the plasma
waves and remain valid for a general bulk Hamiltonian symbol $H(m)$
in continuous media with 1D inhomogeneity. When there is no degeneracy
for a given eigenfrequency, all eigenvectors corresponding to the
eigenfrequency form a 1D complex vector space. 
\begin{defn}
Let $Q\subset M$ be a subset of the parameter space that forms a
manifold with or without boundary. If the $j$-th eigenmode is not
degenerate over $Q$, then the space of disjointed union of all eigenvectors
of $\psi_{j}(q)$ at all $q\in Q$ forms a 1D complex line bundle
\begin{equation}
\pi_{j}:E_{j}\rightarrow Q
\end{equation}
over $Q$. With the standard Hermitian form 
\begin{equation}
\left\langle u,w\right\rangle \equiv u\cdot\bar{w},
\end{equation}
for all $u,w\in\pi_{j}^{-1}(q)$ and $q\in Q,$ $E_{j}\rightarrow Q$
is also a Hermitian bundle. It will be called the Hermitian line bundle
of the $j$-th eigenmode of the bulk Hamiltonian symbol $H(m)$ over
$Q$. 
\end{defn}

If both the $l$-th eigenmode and the $j$-th eigenmode are non-degenerate
over $Q,$ the Whitney sum of $E_{l}\rightarrow Q$ and $E_{j}\rightarrow Q$
defines the Hermitian line bundle of the $l$-th and the $j$-th eigenmodes,
\begin{equation}
E_{\{l,j\}}\equiv E_{l}\oplus E_{j}
\end{equation}
with the Hermitian form defined as 
\begin{equation}
\left\langle \boldsymbol{u},\boldsymbol{v}\right\rangle \equiv u_{l}\cdot\bar{w}_{l}+u_{j}\cdot\bar{w}_{j},
\end{equation}
for all $\boldsymbol{u}=(u_{l},u_{j}),\boldsymbol{w}=(w_{l},w_{j})$,
where $u_{l},w_{l}\in\pi_{l}^{-1}(q)$ and $u_{j},w_{j}\in\pi_{j}^{-1}(q)$.
Similarly, Hermitian bundle of a set of eigenmodes indexed by set
$J$ is defined as 
\begin{equation}
E_{J}\equiv\oplus_{j\in J}E_{j}\thinspace,
\end{equation}
if for each $j\in J,$ the $j$-th eigenmode is not degenerate over
$Q$. In general, $E_{J}$ can be defined when degeneracy exists only
between indices in $J$, but we will not use this structure in the
present study.

The current study is concerned with the topology of the Hermitian
line bundles $E_{j}\rightarrow Q$. In particular, we would like to
know when the bundle is trivial, i.e., a global product bundle over
$Q$, and when it is not. If nontrivial, it is desirable to calculate
the Chern classes of the bundle to measure how twisted it is. For
the Hermitian line bundles of eigenmodes of plasma waves, we will
show in Secs.\,\ref{sec:TLCWPrediction} and \ref{sec:AnalyticalTDC}
that the topological index of the $E_{1}$ bundle over a properly
chosen non-contractible, compact manifold in phase space $Q$, calculated
from its first Chern class $C_{1}(E_{1})$, determines the number
of TLCWs at the transition region.

For Hermitian bundles, the associated principal bundles are $U(n)$
bundles and each Chern class $C_{j}$ is a de Rham cohomology class
of the base manifold constructed from a curvature 2-form of the bundles.
According to the Chern-Weil theorem, different connections for the
bundles yield the same de Rham cohomology classes on the base manifold.
In the present study, it is only necessary to calculate the first
Chern class, and the following result is useful, 
\begin{equation}
C_{1}(E_{J})=\sum_{j\in J}C_{1}(E_{j})\,.\label{eq:C1}
\end{equation}
The right-hand side of Eq.\,(\ref{eq:C1}) is relatively easy to
calculate because each $E_{j}$ is a Hermitian line bundle, whose
first Chern class is given by 
\begin{align}
C_{1} & =\dfrac{\mathrm{i}}{2\pi}\theta,\\
\theta & =\mathrm{d}\chi,
\end{align}
where $\theta$ is a curvature 2-form and $\chi$ is a connection.
As mentioned above, different connections will generate the same $C_{1}$
class. Nevertheless, the Hermitian line bundle is endowed with natural
connection 
\begin{equation}
\chi=\left\langle w,\mathrm{d}w\right\rangle ,
\end{equation}
which is a $u(1)$-valued local 1-form in each trivialization patch.
This natural connection for the Hermitian line bundle is known as
the Berry connection in condensed matter physics or the Simon connection
\footnote{It was Barry Simon \citep{Simon1983,Castelvecchi2020} who first pointed
out that Michael Berry's phase is an anholonomy of the natural connection
on a Hermitian line bundle. For this reason, Frankel \citep{Frankel2011}
taunted the temptation to call it the Berry-Barry connection. However,
as Barry Simon pointed out, it is had been known to geometers such
as Bott and Chern \citep{Bott1965}. Given the current culture of
inclusion, it is probably more appropriate to call it the Bott-Chern-Berry-Simon
connection.}.

\subsection{Trivial topology of plasma waves in momentum space \label{subsec:TrivialTop}}

In terms of topological properties, there is a major difference between
condensed matters and classical continuous media such as plasmas and
fluids. The momentum space, or wavenumber space, of typical condensed
matters is the Brillouin zone, which is non-contractible due to the
periodicity of the lattices. On the contrary, the wavenumber space
in plasmas and fluids is contractible, and it is a well-known fact
that vector bundles over a contractible manifold are trivial. Here,
a topological manifold $M$ is called contractible if it is of the
same homotopy type of a point, i.e., there exist a point $x_{0}$
and continuous maps $f:M$$\rightarrow\{x_{0}\}$ and $g:\{x_{0}\}\rightarrow M$
such that $f\circ g$ is homotopic to identity in $\{x_{0}\}$ and
$g\circ f$ is homotopic to identity in $M.$ Because of its importance
to the continuous media in classical physics, we formalize this result
as a theorem. 
\begin{thm}
\label{thm:contractible}Let $Q$ be a subset of the parameter space
$M$ for a bulk Hamiltonian symbol $H$. If the $j$-th eigenmode
is non-degenerate on $Q$, and $Q$ is a contractible manifold, then
the Hermitian line bundle $E_{j}\rightarrow Q$ is trivial. In particular,
the n-th Chern class $C_{n}($$E_{j}\rightarrow Q)=0$ for $n\ge1.$ 
\end{thm}

Note that Theorem \ref{thm:contractible} holds for any bulk Hamiltonian
symbol. For the $H(\boldsymbol{r},\boldsymbol{k})$ defined in Eq.\,(\ref{eq:H})
and the $H(x,k_{x},k_{y},k_{z})$ defined in Eq.\,(\ref{eq:Hx})
for plasma waves, a more specific result is available as a direct
corollary of Theorem \ref{thm:contractible}. 
\begin{thm}
\label{thm:ContractiblePlasma}For the bulk Hamiltonian symbol $H(\boldsymbol{r},\boldsymbol{k})$
defined in Eq.\,(\ref{eq:H}) for plasma waves, when $k_{z}\neq k^{\pm},$
the Hermitian line bundle of all eigenmodes over the perpendicular
wavenumber plane $Q_{k_{\perp}}=\{(k_{x},k_{y})\mid k_{x}\in\mathbb{R},\thinspace k_{y}\in\mathbb{R}\}=\mathbb{R}^{2}$
are trivial. In particular, $C_{n}(E_{j}\rightarrow Q_{k_{\perp}})=0$
for $n\ge1$ and $-4\le j\le4.$ 
\end{thm}

The fact that $C_{n}(E_{j}\rightarrow Q)$ vanishes when $Q$ is contractible
is expected after all because $C_{n}$ is the de Rham cohomology class
of the base manifold. But Theorem \ref{thm:ContractiblePlasma} is
important. It tells us that the plasma wave topology over the $k_{x}$-$k_{y}$
plane is trivial if $k_{z}\neq k^{\pm}.$ Nontrivial topology of plasma
wave bundles occurs only over non-contractible parameter manifolds,
for example, over an $S^{2}$ surface in the phase space of $(x,k_{x},k_{y})$,
as we will show in Sec.\,\ref{subsec:NonTrivialTop}.

Before leaving this subsection, we would like to point out that in
recent studies of wave topology in classical continuous media, much
effort has been made to calculate ``topological indices'' or ``Chern
numbers'' of the eigenmode bundles over the contractible $k_{x}$-$k_{y}$
plane. This type of effort is characterized by the attempt to evaluate
the integration of the Berry curvature or various modified versions
thereof over the $k_{x}$-$k_{y}$ plane. The difficulties involved
were often attributed to the fact that the $k_{x}$-$k_{y}$ plane
is not compact. As we see from Theorems \ref{thm:contractible} and
\ref{thm:ContractiblePlasma}, when the wave bundle is well-defined
over the entire $k_{x}$-$k_{y}$ plane, its topology is trivial.
In these cases, the non-compactness of the $k_{x}$-$k_{y}$ plane
is irrelevant, so is whether an integer or non-integer index can be
designed.

\subsection{Nontrivial plasma wave topology in phase space \label{subsec:NonTrivialTop}}

In the present context, the ultimate utility of the topological property
of the eigenmode bundles of the bulk Hamiltonian symbol $H$ is to
predict the existence of the topological edge modes of the global
Hamiltonian PDO $\hat{H}.$ For this purpose, the proper plasma wave
eigenmode bundles are over the 2D sphere in the parameter space of
$(x,k_{x},k_{y})$ for the $H(x,k_{x},k_{y},k_{z})$ defined in Eq.\,(\ref{eq:Hx}),
\begin{equation}
S_{1}^{2}=\left\{ (x,k_{x},k_{y})\mid x^{2}+k_{x}^{2}+k_{y}^{2}=1\right\} =\partial B_{1}^{3},
\end{equation}
where $B_{1}^{3}$ is the 3D ball with radius $r=1$.

One indicator of nontrivial topology, or twist, of a wave eigenmode
bundle over $S_{1}^{2}$ is the number of zeros a nontrivial section
must have, akin to the situation of hairy ball theorem for the tangent
bundle of $S^{2}$. To include the possibilities of repeated zeros,
we follow Frankel \citep{Frankel2011} to define the index of an isolated
zero point $z$ of a section $u$ of a Hermitian line bundle as follows. 
\begin{defn}
\label{def:ZeroIndex}Let $z$ be an isolated zero of a section $u$
that has a finite number of zeros. Select a normalized local frame
$e$ for the Hermitian line bundle in the neighborhood of $z$. The
normalized section near $z$, but not at $z,$ can be expressed as
$u/\mid u\mid=e\text{exp}(\mathrm{i}\alpha).$ The index at $z$ is
defined to be 
\begin{equation}
j_{u}(z)\equiv\frac{1}{2\pi}\int_{\partial D}\mathrm{d}\alpha,
\end{equation}
where $D$ is a small disk containing $z$ on $S_{1}^{2}$ with orientation
pointing away from $S_{1}^{2},$ and the orientation of $\partial D$
is induced from that of $D$. The index of the section $u$ is the
sum of indices at all zeros $z_{l}$ of $u,$ 
\begin{equation}
\text{Ind}(u)\equiv\sum_{l}j_{u}(z_{l}).
\end{equation}
\end{defn}

Note that for each isolated zero $z$, the local frame $e$ selected
in the neighborhood of $z$ is not vanishing at $z$. The index $j_{u}(z)$
intuitively measures how many turns the phase of $u$ increases relative
to $e$ at $z$ over one turn on $\partial D$. In general, $e$ is
only a local frame instead of a global frame, otherwise the bundle
is trivial.

The following theorem of Chern relates the index of a nontrivial section
to the first Chern class over $S^{2}$ \citep{Frankel2011}. 
\begin{thm}
\label{thm:Chern} {[}Chern{]} Let $E$ be a Hermitian line bundle
over a closed orientable 2D surface $S^{2}$. Let $u:$$S^{2}\rightarrow E$
be a section of $E$ with a finite number of zeros. Then the integral
of the 1st Chern class $C_{1}$ over $S^{2}$ is an integer that is
equal to the index of the section $u,$ i.e., 
\[
n_{c}\equiv\int_{S^{2}}C_{1}(E\rightarrow S^{2})=\mathrm{Ind}(u).
\]
\end{thm}

Here, $n_{c}$ is known as the first Chern number. Since the present
study only involves the first Chern number, it is denoted by $n_{c}$
instead of $n_{c_{1}}.$ The first Chern number of the $j$-th eigenmode
bundle over $S_{1}^{2}$ is denoted by $n_{cj}$. In Sec.\,\ref{sec:TLCWPrediction},
the first Chern number $n_{c1}$ of the plasma wave eigenmode bundle
$E_{1}$ over $S_{1}^{2}$ will be linked to the spectral flow index
of $\hat{H}(x,-\mathrm{i}\eta\partial_{x},k_{y},k_{z})$ using Faure's
index theorem \citep{Faure2019}. To facilitate the calculation of
$n_{cj}$ over $S_{1}^{2}$, we establish the following general facts
about eigenmode bundles in continuous media, including plasma wave
eigenmode bundles. 
\begin{defn}
\label{def:Iso2} Let $\pi_{1}:E_{1}\rightarrow P$ and $\pi_{2}:E_{2}\rightarrow Q$
be two vector bundles. A diffeomorphism $\phi:E_{1}\rightarrow E_{2}$
is called an isomorphism if for every $p\in P$, $\exists$ $q\in Q$
such that $\phi\circ\pi_{1}^{-1}(p)\subset\pi_{2}^{-1}(q)$ and $\phi:\pi_{1}^{-1}(p)\rightarrow\pi_{2}^{-1}(q)$
is a vector space isomorphism. If an isomorphism exists, $E_{1}$
and $E_{2}$ are isomorphic, denoted as $E_{1}\simeq E_{2}.$ 
\end{defn}

Note that the base manifolds $P$ and $Q$ in the above definition
can be identical or different. 
\begin{defn}
Let $f:P\rightarrow Q$ is a smooth map between differential manifolds
$P$ and $Q,$ and $\pi_{2}:E\rightarrow Q$ a vector bundle over
$Q$. The pullback bundle $\pi_{1}:f^{*}E\rightarrow P$ is defined
to be 
\begin{equation}
f^{*}E=\left\{ (p,e)\in P\times E\mid p\in P,\,e\in E,\,f(p)=\pi_{2}(e)\right\} .
\end{equation}
\end{defn}

The standard manifold and vector bundle structure of $f^{*}E$ can
be formally established. For example, see Ref. \citep{Tu2017-177}.
Note that the pullback bundle is defined as a pullback set and this
mechanism does not define a map from $E$ to $f^{*}E$ because $f$
is not in general invertible. On the other hand, when $f$ is invertible,
a pullback map from $E$ to $f^{*}E$ can be defined by a similar
mechanism as follows. 
\begin{defn}
Let $f:P\rightarrow Q$ is a diffeomorphism between differential manifolds
$P$ and $Q,$ and $\pi_{2}:E\rightarrow Q$ a vector bundle over
$Q$. The pullback map $f^{\dagger}$ is defined to be 
\begin{align}
f^{\dagger}: & E\rightarrow f^{*}E,\\
 & e\mapsto(f^{-1}\circ\pi_{2}(e),e).\nonumber 
\end{align}
When $f$ is a diffeomorphism, $f^{\dagger}(E)=f^{*}E.$ 
\end{defn}

We will use the following theorem known as homotopy induced isomorphism
\citep{Bott1982}. 
\begin{thm}
\label{thm:homotopyIso}{[}Homotopy induced isomorphism{]} Let $f_{0}$
and $f_{1}$ are two homotopic maps between manifolds $P$ and $Q$.
For a vector bundle $E\rightarrow Q$, the pullback bundles $f_{0}^{*}E$
and $f_{1}^{*}E$ over $P$ are isomorphic. 
\end{thm}

Theorem \ref{thm:contractible} is a direct corollary of Theorem \ref{thm:homotopyIso}. 
\begin{thm}
\label{thm:pullbackIso} Let $f:P\rightarrow Q$ is a diffeomorphism
between differential manifolds $P$ and $Q$, and $\pi_{2}:E\rightarrow Q$
a vector bundle over $Q$. The pullback map $f^{\dagger}$ is an isomorphism
between $E$ and $f^{*}E$. 
\end{thm}

\begin{proof}
According to Definition \ref{def:Iso2}, it suffices to prove that
(i) $f^{\dagger}:E\rightarrow f^{*}E$ is a diffeomorphism and (ii)
$\forall q\in Q,$ $\exists p\in P$ such that $f^{\dagger}\circ\pi_{2}^{-1}(q)\subset\pi_{1}^{-1}(p)$
and $f^{\dagger}:\pi_{2}^{-1}(q)\mapsto\pi_{1}^{-1}(p)$ is an isomorphism
of vector space.

Since $f^{\dagger}(e)=(f^{-1}\circ\pi_{2}(e),e)$ and both $f^{-1}$
and $\pi_{2}$ are smooth, $f^{\dagger}$ is smooth. By construction,
$f^{\dagger}$ is smoothly invertible. Thus, $f^{\dagger}$ is a diffeomorphism.
To prove (ii), we utilize a local trivialization. For $\forall q\in Q$,
let $U$ be an open set containing $q$ in the open cover of $Q$
for the local trivialization of $E\rightarrow Q$. Locally, $E$ is
a production $U\times V.$ In particular, $\pi_{2}^{-1}(q)=\{q\}\times V$
and 
\begin{equation}
f^{\dagger}:(q,v)\mapsto\left(p=f^{-1}\circ\pi_{2}(q,v)=f^{-1}(q),(q,v)\right).
\end{equation}
For this fixed $q,$ 
\begin{align}
f^{\dagger}\circ\pi_{2}^{-1}(q) & =f^{\dagger}\left(\left\{ (q,v)\mid v\in V\right\} \right)=\left\{ \left(p,(q,v)\right)\mid p=f^{-1}(q),v\in V\right\} \nonumber \\
 & =\left\{ \left(p,(q,v)\right)\mid f(p)=\pi_{2}(q,v),v\in V\right\} =\pi_{1}^{-1}(p).
\end{align}
Also, $f^{\dagger}:(q,v)\mapsto\left(p,(q,v)\right)$ for the fixed
$q$ and $p=f^{-1}(q)$ is an isomorphism. Thus, $f^{\dagger}$ is
an isomorphism and $E\simeq f^{*}E$. 
\end{proof}
The following is the main theorem of this paper, which will enable
us to analytical calculate the topological index for the TLCW. In
a parameter space that is $\mathbb{R}^{3}$, denote by $S_{r}^{2}=\left\{ q=(q_{1,}q_{2,}q_{3})\mid q_{1}^{2}+q_{2}^{2}+q_{3}^{2}=r^{2}\right\} $
the sphere of radius $r$, and by $Sh_{(a,b)}\equiv\left\{ q=(q_{1,}q_{2,}q_{3})\mid a^{2}\le q_{1}^{2}+q_{2}^{2}+q_{3}^{2}\le b^{2}\right\} $
the 3D shell with inner radius $a$ and outer radius $b$. 
\begin{thm}
\label{thm:BoundaryIso}{[}Boundary isomorphism{]} Let $E\rightarrow Sh_{(a,b)}$
be a Hermitian line bundle defined over a shell $Sh_{(a,b)}$ in a
parameter space that is $\mathbb{R}^{3}$.

(i) The bundles obtained by restricting $E$ over $S_{a}^{2}$ and
$S_{b}^{2}$ are isomorphic, i.e., $E\rightarrow S_{a}^{2}\simeq E\rightarrow S_{b}^{2}$.

(ii) Bundles $E\rightarrow S_{a}^{2}$ and $E\rightarrow S_{b}^{2}$
have the same first Chern number, i.e., $n_{c}\left(E\rightarrow S_{a}^{2}\right)=n_{c}\left(E\rightarrow S_{b}^{2}\right)$. 
\end{thm}

\begin{proof}
To prove (i), construct the following continuous class of compressing
maps on $Sh_{(a,b)}$, 
\begin{align}
f_{\epsilon}: & Sh_{(a,b)}\rightarrow Sh_{(a,b)}\thinspace,\\
 & p\mapsto\epsilon p+(1-\epsilon)\frac{ap}{\left|p\right|},\thinspace0\le\epsilon\le1.\nonumber 
\end{align}
As $\epsilon$ decreases from $1$ to $0$ continuously, $f_{\epsilon}$
compresses the shell towards the inner sphere $S_{a}^{2}$. $f_{1}$
is the identify map, $f_{0}$ crashes the shell onto $S_{a}^{2},$
and $f_{1}$ and $f_{0}$ are homotopic. According to Theorem \ref{thm:homotopyIso},
\begin{equation}
f_{1}^{*}E\simeq f_{0}^{*}E\:.\label{eq:f1E}
\end{equation}
Restricting both sides of Eq.\,(\ref{eq:f1E}) to $S_{b}^{2}$ leads
to 
\begin{equation}
E\mid_{S_{b}^{2}}=\left(f_{1}^{*}E\right)\mid_{S_{b}^{2}}\simeq\left(f_{0}^{*}E\right)\mid_{S_{b}^{2}}=f_{0}^{*}\left(E\mid_{S_{a}^{2}}\right)\mid_{S_{b}^{2}}.
\end{equation}
Denote by $f_{0r}$ the restriction of $f_{0}$ on $S_{b}^{2},$ i.e.,
\begin{align}
f_{0r}: & S_{b}^{2}\rightarrow S_{a}^{2}\thinspace,\nonumber \\
 & p\mapsto f_{0}(p).
\end{align}
Obviously, $f_{0r}$ is a diffeomorphism, and according to Theorem
\ref{thm:pullbackIso}, 
\begin{equation}
f_{0}^{*}\left(E\mid_{S_{a}^{2}}\right)\mid_{S_{b}^{2}}=f_{0r}^{*}\left(E\mid_{S_{a}^{2}}\right)\simeq E\mid_{S_{a}^{2}}.
\end{equation}
Therefore, 
\begin{equation}
E\rightarrow S_{b}^{2}=E\mid_{S_{b}^{2}}\simeq E\mid_{S_{a}^{2}}=E\rightarrow S_{a}^{2}\,.
\end{equation}

For (ii), we have 
\[
f_{0r}^{*}\left(C_{1}\left(E\rightarrow S_{a}^{2}\right)\right)=C_{1}\left(f_{0r}^{*}\left(E\rightarrow S_{a}^{2}\right)\right)=C_{1}\left(E\rightarrow S_{b}^{2}\right),
\]
where the first equal sign is the naturality property of characteristic
classes. The second equal sign is due to the fact that $f_{0r}^{*}\left(E\rightarrow S_{a}^{2}\right)$
and $E\rightarrow S_{b}^{2}$ are two isomorphic bundles on $S_{b}^{2}$,
and thus have the same Chern classes. The first Chern number on $E\rightarrow S_{b}^{2}$
is 
\begin{align}
n_{c}\left(E\rightarrow S_{b}^{2}\right) & =\int_{S_{b}^{2}}C_{1}\left(E\rightarrow S_{b}^{2}\right)\\
 & =\int_{S_{b}^{2}}f_{0r}^{*}\left(C_{1}\left(E\rightarrow S_{a}^{2}\right)\right)\nonumber \\
 & =\int_{S_{a}^{2}}C_{1}\left(E\rightarrow S_{a}^{2}\right)=n_{c}\left(E\rightarrow S_{a}^{2}\right),\nonumber 
\end{align}
where the integral on $S_{b}^{2}$ is evaluated on $S_{a}^{2}$ via
the pullback mechanism in the third equal sign. 
\end{proof}
Theorem \ref{thm:BoundaryIso} says that when the Hermitian line bundle
$E$ is defined on $Sh(a,b),$ $E\rightarrow S_{a}^{2}$ and $E\rightarrow S_{b}^{2}$
are isomorphic bundles and have the same first Chern number. Around
an isolated Weyl point, the Hermitian line bundle is well defined
except at the Weyl point, Theorem \ref{thm:BoundaryIso} states that
all closed surfaces surrounding the Weyl point have the same first
Chern number, which can be viewed as the topological charge associated
with this isolated Weyl point in phase space (see Fig.\,\ref{fig:TopologicalCharge}).

\begin{figure}[ht]
\includegraphics[width=6cm]{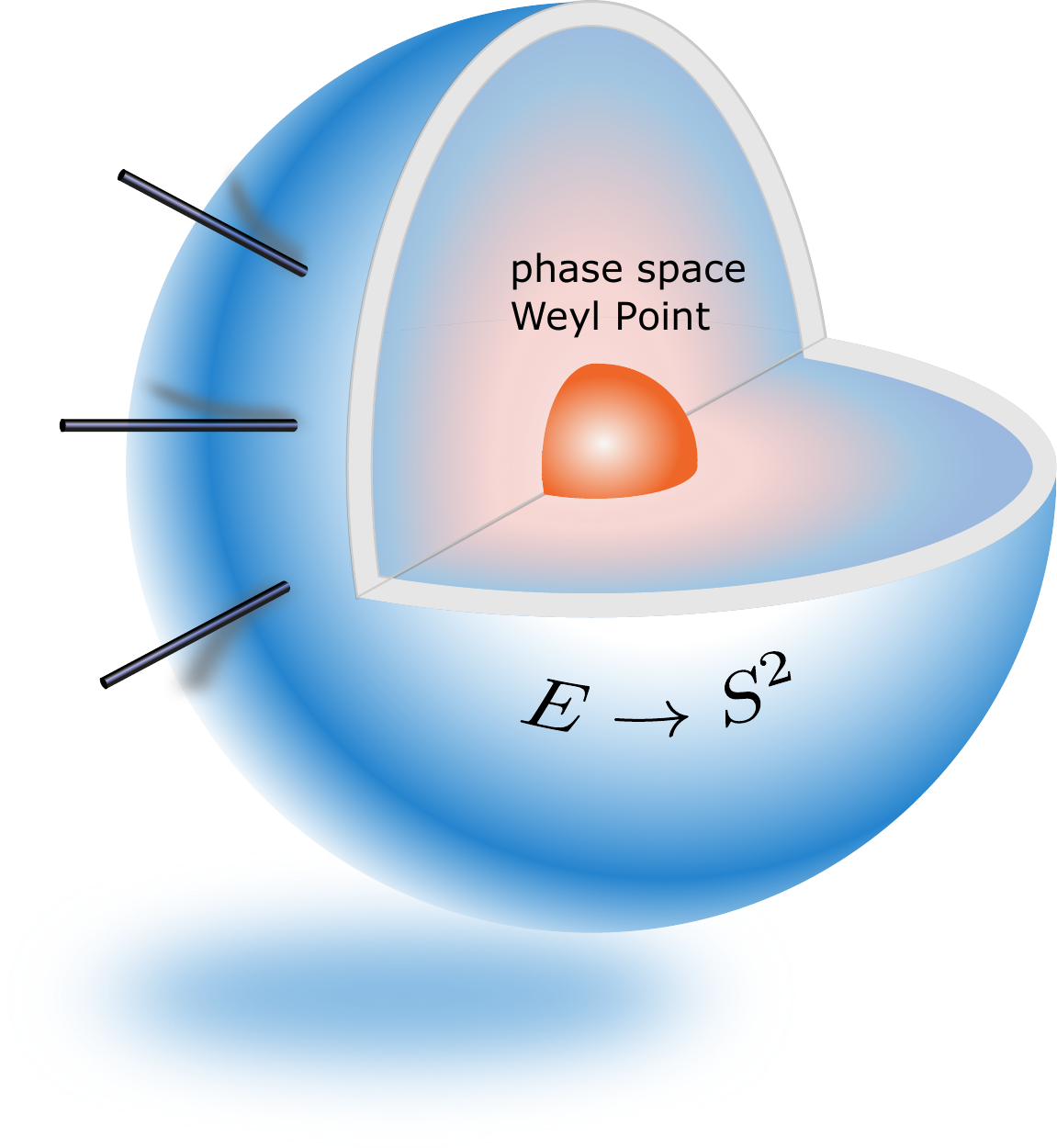} \caption{Topological charge of waves in classical continuous media. Around
an isolated Weyl point, all closed surfaces surrounding the Weyl point
have the same first Chern number, which can be viewed as the topological
charge associated with this isolated Weyl point in phase space. }
\label{fig:TopologicalCharge}
\end{figure}

We also need the following results for the first Chern class for the
plasma wave eigenmode bundles. They can be established straightforwardly. 
\begin{lem}
\label{lem:C1Property}For the 9 bulk eigenmode bundles of the plasma
waves specified by Hamiltonian symbol $H(\boldsymbol{r},\boldsymbol{k})$
defined in Eq.\,(\ref{eq:H}), the following identities for the first
Chern class holds over a general base manifold for the bundles: 
\begin{flalign}
C_{1}\left(\oplus_{j=1}^{4}E_{j}\right) & =C_{1}\left(\oplus_{j=-4}^{-1}E_{j}\right),\\
C_{1}\left(\oplus_{j=-4}^{4}E_{j}\right) & =0,\\
C_{1}\left(E_{0}\right) & =0,\\
C_{1}\left(\oplus_{j=-4}^{1}E_{j}\right) & =C_{1}\left(E_{1}\right).
\end{flalign}
\end{lem}

\section{TLCW predicted by Faure's index theorem and algebraic topological
analysis\label{sec:TLCWPrediction}}

\subsection{Faure's Index theorem for TLCW}

In condensed matter physics, the bulk-edge correspondence states that
the gap Chern number equals the number of edge modes in the gap. Mathematically,
the correspondence had been rigorously proved as the Atiyah-Patodi-Singer
(APS) index theorem \citep{Atiyah1976} for spectral flows over $S^{1}$,
which corresponds to the momentum parameter $k_{y}$ in the direction
with spatial translation symmetry of a periodic lattice. However,
for waves in classical continuous media, including waves in plasmas,
the $k_{y}$ parameter is not periodic, and it takes value in $\mathbb{R}$.
Therefore, the APS index theorem proved for spectral flow over $S^{1}$
is not applicable for waves in continuous media without modification.
Recently, Faure \citep{Faure2019} formulated an index theorem for
spectral flows over $\mathbb{R}$-valued $k_{y}$, which links the
spectral flow index to the gap Chern number of the eigenmode bundle
over a 3D ball in the phase space of $(x,k_{x},k_{y}).$ Faure's index
theorem applies to waves in classical continuous media. In this section,
we apply Faure's index theorem and Theorem \ref{thm:BoundaryIso}
to prove the existence of TLCW. For the bulk Hamiltonian symbol $H(x,k_{x},k_{y},k_{z})$
defined in Eq.\,(\ref{eq:Hx}), the global Hamiltonian PDO $\hat{H}(x,-\mathrm{i}\eta\partial_{x},k_{y},k_{z})$
defined in Eq.\,(\ref{eq:Hhatx}), and the 1D equilibrium profile
specified by Eq.\,(\ref{eq: condition1}), we have the following
theorems and definition adapted from Faure \citep{Faure2019}. 
\begin{thm}
\label{thm:SpectralFlow}For a fixed $k_{z}$, assume that $[g_{1},g_{2}]$
is the common gap of $\omega_{1}(x,k_{x},k_{y})$ and $\omega_{2}(x,k_{x},k_{y})$
for parameters exterior to the ball $B_{1}^{3}$. For any $\lambda>0,$
there exists $\eta_{0}>0$ such that

(i) for all $\eta<\eta_{0}$ and $k_{y}\in[-1-\lambda,1+\lambda]$,
$\hat{H}(x,-\mathrm{i}\eta\partial_{x},k_{y},k_{z})$ has no or discrete
spectrum in the gap of $[g_{1}+\lambda,g_{2}-\lambda]$ that depend
on $\eta$ and $k_{y}$ continuously;

(ii) for all $\eta<\eta_{0}$, $\hat{H}(x,-\mathrm{i}\eta\partial_{x},k_{y},k_{z})$
has no spectrum in $[g_{1}-\lambda,g_{2}+\lambda]$ at $k_{y}=\pm(1+\lambda).$ 
\end{thm}

\begin{proof}
This theorem is a special case of Theorem 2.2 in Faure \citep{Faure2019}. 
\end{proof}
Theorem \ref{thm:SpectralFlow} states that the spectrum of $\hat{H}(x,-\mathrm{i}\eta\partial_{x},k_{y},k_{z})$
in the common gap $[g_{1}+\lambda,g_{2}-\lambda]$, if any, must consist
of discrete dispersion curves parameterized by $k_{y}$. Theorem \ref{thm:SpectralFlow}
also stipulates the following ``traffic rules'' for the flow of
the spectrum. The dispersion curves cannot enter or exit the rectangle
region $[-1-\lambda,1+\lambda]\times[g_{1}+\lambda,g_{2}-\lambda]$
on the $k_{y}$-$\omega$ plane from the left or right sides. They
can only enter or exit through the upper or lower sides (see Fig.\,\ref{fig:SpectralFlow}).
Intuitively, a spectral flow is a dispersion curve of $\hat{H}(x,-\mathrm{i}\eta\partial_{x},k_{y},k_{z})$
that can trespass the rectangle. It flows between the lower band and
the upper band, as if transporting one eigenmode upward or downward
through the spectral gap of $H(x,k_{x},k_{y},k_{z})$ for parameters
exterior to the ball $B_{1}^{3}$. We now formally define the spectral
flow and spectral flow index.

\begin{figure}[ht]
\centering \includegraphics[width=10cm]{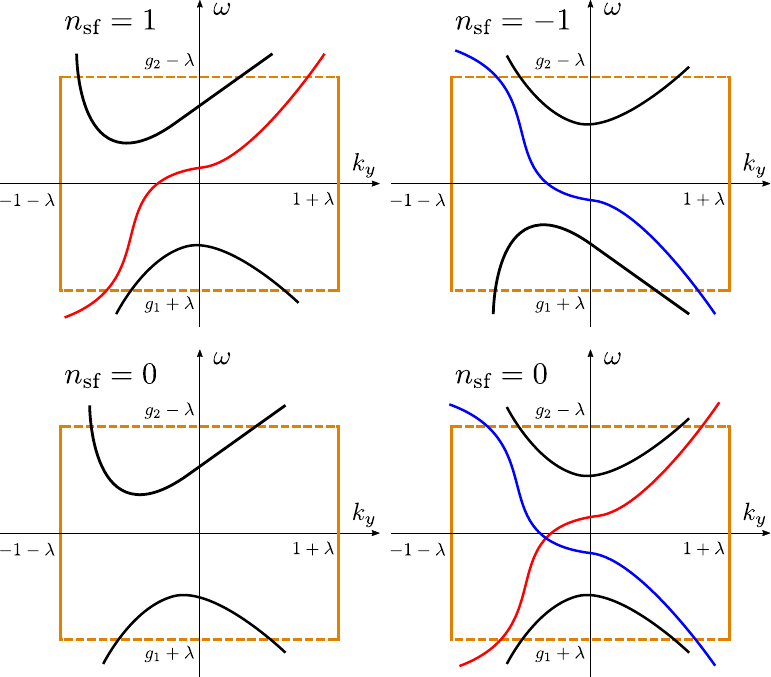} \caption{Illustration of possible spectral flows of $\hat{H}(x,-\mathrm{i}\eta\partial_{x},k_{y},k_{z})$
in the common gap $[g_{1}+\lambda,g_{2}-\lambda]$. Theorem \ref{thm:SpectralFlow}
stipulates the ``traffic rules'' for the flow of the spectrum. Red
curves are spectral flows with index $1$, and blue curves are spectral
flows with index $-1.$}
\label{fig:SpectralFlow}
\end{figure}

\begin{defn}
\label{def:SF}For a fixed $k_{z}$, assume that $[g_{1},g_{2}]$
is the common gap of $\omega_{1}(x,k_{x},k_{y})$ and $\omega_{2}(x,k_{x},k_{y})$
for parameters exterior to the ball $B_{1}^{3}$. A spectral flow
is a smooth dispersion curve $\omega=f(k_{y},\eta)$ of $\hat{H}(x,-\mathrm{i}\eta\partial_{x},k_{y},k_{z})$
satisfying the following condition: For any $\lambda>0,$ there exists
a $\eta_{0}>0$ such that either (i) for all $\eta<\eta_{0}$, $f(-1-\lambda,\eta)<g_{1}+\lambda$
and $f(1+\lambda,\eta)>g_{2}-\lambda$ or (ii) for all $\eta<\eta_{0}$,
$f(-1-\lambda,\eta)>g_{2}-\lambda$ and $f(1+\lambda,\eta)<g_{1}+\lambda$.
For case (i), its index is $1$. For case (ii), its index is $-1$.
The spectral flow index $n_{\text{sf}}$ of $\hat{H}(x,-\mathrm{i}\eta\partial_{x},k_{y},k_{z})$
is the summation of indices of all its spectral flows.

In the present context, a spectral flow of $\hat{H}(x,-\mathrm{i}\eta\partial_{x},k_{y},k_{z})$
is a TLCW. But Theorem \ref{thm:SpectralFlow} and Definition \ref{def:SF}
are valid for any generic $\hat{H}$. 
\end{defn}

A few possible spectral flow configurations are illustrated in Fig.\,\ref{fig:SpectralFlow}.
Strictly speaking, $n_{\text{sf}}$ is not necessarily the total number
of all possible upward and downward spectral flows of $\hat{H}(x,-\mathrm{i}\eta\partial_{x},k_{y},k_{z})$.
It is the net number of upward spectral flows.

For plasma waves, the following theorem links the number of TLCWs
to the first Chern number of the $E_{1}$ eigenmode bundle over a
non-contractible, compact surface in the phase space of $(x,k_{x},k_{y}).$ 
\begin{thm}
\label{thm:SpectralFlowIndex}For a fixed $k_{z},$ assume that the
common gap condition for parameters exterior to the ball $B_{1}^{3}\equiv\left\{ (x,k_{x},k_{y})\mid x^{2}+k_{x}^{2}+k_{y}^{2}\le1\right\} $
is satisfied for the spectra $\omega_{1}(x,k_{x},k_{y},k_{z})$ and
$\omega_{2}(x,k_{x},k_{y},k_{z})$ of $H(x,k_{x},k_{y},k_{z})$. The
spectral flow index of $\hat{H}(x,-\mathrm{i}\eta\partial_{x},k_{y},k_{z})$
in the gap equals the Chern number $n_{c}\left(E_{1}\rightarrow S_{1}^{2}\right)$
of the $E_{1}$ eigenmode bundle of $H(x,k_{x},k_{y},k_{z})$ over
$S_{1}^{2}\equiv\left\{ (x,k_{x},k_{y})\mid x^{2}+k_{x}^{2}+k_{y}^{2}\le1\right\} $,
i.e., $n_{\text{sf}}=n_{c}\left(E_{1}\rightarrow S_{1}^{2}\right)$. 
\end{thm}

\begin{proof}
This theorem is a direct specialization of Theorem 2.7 formulated
by Faure in Ref.\,\citep{Faure2019}, which states that when a spectral
gap exists between $\omega_{l}$ and $\omega_{l+1}$ of a bulk Hamiltonian
for all parameters exteriors to $B_{r}^{3}$, the spectral flow index
$n_{\text{sf}}$ of the corresponding PDO in the gap equals the gap
Chern number $n_{c}\left(\oplus_{j\le l}E_{j}\rightarrow S_{r}^{2}\right)$.
For the plasma waves satisfying the common gap condition stated, $l=1$
and 
\[
n_{\text{sf}}=n_{c}\left(\oplus_{j\le1}E_{j}\rightarrow S_{1}^{2}\right)=n_{c}\left(E_{1}\rightarrow S_{1}^{2}\right),
\]
where use is made of Lemma \ref{lem:C1Property}. 
\end{proof}

\subsection{Index calculation of TLCW using algebraic topological techniques }

Theorem \ref{thm:SpectralFlowIndex} links the number of TLCWs, or
the spectral flow index of $\hat{H}(x,-\mathrm{i}\eta\partial_{x},k_{y},k_{z})$,
to the Chern number $n_{c}\left(E_{1}\rightarrow S_{1}^{2}\right)$
of the $E_{1}$ eigenmode bundle of $H(x,k_{x},k_{y},k_{z})$ over
$S_{1}^{2}\equiv\left\{ (x,k_{x},k_{y})\mid x^{2}+k_{x}^{2}+k_{y}^{2}\le1\right\} $.
However, it is not an easy task to calculate $n_{c}\left(E_{1}\rightarrow S_{1}^{2}\right)$
either analytically or numerically. Here, we use the algebraic topological
tools developed in Sec.\,\ref{subsec:NonTrivialTop} to analytically
calculate $n_{c}\left(E_{1}\rightarrow S_{1}^{2}\right).$

Because for the 1D equilibrium profile specified by Eq.\,(\ref{eq: condition1}),
the LC Weyl point only occurs at $x=0,$ and the eigenmode bundle
$E_{1}$ is well-defined in $B_{1}^{3}/(0,0,0)$, we can invoke Theorem
\ref{thm:BoundaryIso} to calculate $n_{c}\left(E_{1}\rightarrow S_{1}^{2}\right)$
as 
\begin{equation}
n_{c}\left(E_{1}\rightarrow S_{1}^{2}\right)=\lim_{\delta\rightarrow0^{+}}n_{c}\left(E_{1}\rightarrow S_{\delta}^{2}\right).\label{eq:ncepsilon}
\end{equation}
The right-hand side of Eq.\,(\ref{eq:ncepsilon}) is the first Chern
number of the $E_{1}$ bundle over an infinitesimal sphere surrounding
the Weyl point in the phase space of $(x,k_{x},k_{y})$, and it can
be analytically evaluated using Taylor expansion at the Weyl point
as follows.

At the LC Weyl point $(x,k_{x},k_{y})=(0,0,0)$, the spectrum and
eigenmodes of $H(x,k_{x},k_{y},k_{z})$ can be solved analytically.
Denote by $\left(\omega_{j0},\psi_{j0}\right)$ the $j$-th eigenmode.
At this point, two of the eigenmodes with positive frequencies resonant,
\begin{equation}
\omega_{10}=\omega_{20}=\omega_{\text{pc}}=\dfrac{\sqrt{k_{z}^{4}+4k_{z}^{2}}-k_{z}^{2}}{2},
\end{equation}
and the corresponding eigenmodes are 
\begin{alignat*}{1}
\psi_{10} & =\left(0,0,-\frac{\mathrm{i}}{\sqrt{2}},0,0,\frac{1}{\sqrt{2}},0,0,0\right)^{\mathrm{T}},\\
\psi_{20} & =\left(\mathrm{i}k_{z},-k_{z},0,\frac{\omega_{\text{pc}}}{k_{z}},\mathrm{i}\frac{\omega_{\text{pc}}}{k_{z}},0,-\mathrm{i},1,0\right)^{\mathrm{T}},
\end{alignat*}
where $\psi_{10}$ is the Langmuir wave and $\psi_{20}$ is the cyclotron
wave. In the infinitesimal neighborhood of the Weyl point, $k_{x}\sim k_{y}\sim x\sim\delta$,

\begin{alignat}{1}
H(x,k_{x},k_{y},k_{z}) & =H_{0}+\delta H,\\
H_{0}(x,k_{x},k_{y},k_{z}) & =\begin{pmatrix}\mathrm{i}\boldsymbol{e}_{z}\times & -\mathrm{i}\omega_{\text{pc}} & 0\\
\mathrm{i}\omega_{\text{pc}} & 0 & (0,0,-k_{z})\times\\
0 & (0,0,k_{z})\times & 0
\end{pmatrix},\\
\delta H(x,k_{x},k_{y},k_{z}) & =\begin{pmatrix}0 & -\mathrm{i}\omega_{\text{p}}^{\prime}(x)x & 0\\
\mathrm{i}\omega_{\text{p}}^{\prime}(x)x & 0 & (-k_{x},-k_{y,}0)\times\\
0 & (k_{x},k_{y,}0)\times & 0
\end{pmatrix}.
\end{alignat}
We can express $H$ in the basis of $\psi_{j0}$ $(-4\le j\le4)$.
But for modes with $\delta\omega=$$\omega-\omega_{\text{pc}}\sim\delta$,
$H$ can approximated by the expansion using $\psi_{10}$ and $\psi_{20}$
only, and reduces to a $2\times2$ matrix, 
\begin{alignat}{1}
H_{2} & (x,k_{x},k_{y},k_{z}):=\begin{pmatrix}\psi_{10}^{\dagger}H\psi_{10} & \psi_{10}^{\dagger}H\psi_{20}\\
\psi_{20}^{\dagger}H\psi_{10} & \psi_{20}^{\dagger}H\psi_{20}
\end{pmatrix}=\begin{pmatrix}\omega_{\text{pc}}+\delta\omega_{\text{p}} & \dfrac{-k_{x}-\mathrm{i}k_{y}}{\sqrt{2}\alpha}\\[10pt]
\dfrac{-k_{x}+\mathrm{i}k_{y}}{\sqrt{2}\alpha} & \omega_{\text{pc}}-\dfrac{4\omega_{\text{pc}}}{\alpha^{2}}\delta\omega_{\text{p}}
\end{pmatrix},\label{eq:H2}\\[5pt]
\alpha & \equiv\sqrt{4+3k_{z}^{2}-k_{z}\sqrt{4+k_{z}^{2}}}\thinspace,\quad\delta\omega_{\text{p}}=-\beta x\thinspace,\quad\beta\equiv\left|\dfrac{\mathrm{d}\omega_{\text{p}}}{\mathrm{d}x}\right|_{x=0}\geq0,
\end{alignat}
where we used the fact that the equilibrium profile $\omega_{\mathrm{p}}(x)$
selected in Eq.\,(\ref{eq: condition1}) decreases monotonically.
The eigen system of $H_{2}$ can be solved straightforwardly. The
two eigenfrequencies of $H_{2}$ are 
\begin{alignat}{1}
\omega_{1} & =\omega_{\text{pc}}-\frac{\beta}{2}\left(1-\frac{4\omega_{\text{pc}}}{\alpha^{2}}\right)x-\gamma\thinspace,\label{eq:om1}\\
\omega_{2} & =\omega_{\text{pc}}-\frac{\beta}{2}\left(1-\frac{4\omega_{\text{pc}}}{\alpha^{2}}\right)x+\gamma\thinspace,\label{eq:om2}\\
\gamma & \equiv\sqrt{\frac{k_{x}^{2}+k_{y}^{2}}{2\alpha^{2}}+\frac{x^{2}\beta^{2}}{4}\left(1+\frac{4\omega_{\text{pc}}}{\alpha^{2}}\right)^{2}}\thinspace.
\end{alignat}
The corresponding eigenmodes, expressed in the basis of $\psi_{10}$
and $\psi_{20}$ , are 
\begin{alignat}{1}
\tilde{\psi}_{1} & =\left(\alpha\beta\left(1+\frac{4\omega_{\text{pc}}}{\alpha^{2}}\right)x+2\alpha\gamma,\sqrt{2}(k_{x}-\mathrm{i}k_{y})\right)^{\mathrm{T}}\thinspace,\\
\tilde{\psi}_{2} & =\left(\alpha\beta\left(1+\frac{4\omega_{\text{pc}}}{\alpha^{2}}\right)x-2\alpha\gamma,\sqrt{2}(k_{x}-\mathrm{i}k_{y})\right)^{\mathrm{T}}.
\end{alignat}
Everywhere except $(x,k_{x},k_{y})=(0,0,0)$ in the parameter space,
we have $\omega_{1}<\omega_{2}$, so the $E_{1}$ eigenmode bundle
of $H$ is faithfully represented by $\tilde{\psi}_{1}$ when $\delta$
is small but non-vanishing. What matters for the present study is
the first Chern number $n_{c}\left(E_{1}\rightarrow S_{\delta}^{2}\right),$
which can be obtained by counting the number of zeros of $\tilde{\psi}_{1}$
on $S_{\delta}^{2}$, according to Theorem \ref{thm:Chern}.

On $S_{\delta}^{2}$, $\tilde{\psi}_{1}$ is well-defined everywhere,
and has one zero at $(x,k_{x},k_{y})=(-\delta,0,0)$. The index of
this zero can be calculated according to Definition \ref{def:ZeroIndex}
as follows. We select the following local frame for $E_{1}\rightarrow S_{\delta}^{2}$
in the neighborhood of $(x,k_{x},k_{y})=(-\delta,0,0)$, 
\begin{equation}
e=\left(\frac{\alpha\beta\left(1+\frac{4\omega_{\text{pc}}}{\alpha^{2}}\right)x+2\alpha\gamma}{(k_{x}-\mathrm{i}k_{y})},\sqrt{2}\right)^{\mathrm{T}}\thinspace.
\end{equation}
It is easy to verify that $e$ is well-defined in the neighborhood
of $(x,k_{x},k_{y})=(-\delta,0,0)$ on $S_{\delta}^{2}$, especially
at the point of $(x,k_{x},k_{y})=(-\delta,0,0)$ itself. Note that
$e$ is singular at $(x,k_{x},k_{y})=(\delta,0,0)$ on $S_{\delta}^{2}$,
therefore it is not a (global) section of bundle $E_{1}\rightarrow S_{\delta}^{2}$.
The expression of the section $\tilde{\psi}_{1}$ in the $e$ frame
is $(k_{x}-\mathrm{i}k_{y})$. In one turn on $S_{\delta}^{2}$ circulating
$(x,k_{x},k_{y})=(-\delta,0,0)$, for example on a circle with a fixed
$x$ near $(x,k_{x},k_{y})=(-\delta,0,0)$, the phase increase of
$(k_{x}-\mathrm{i}k_{y})$ is $2\pi.$ Thus, we conclude that $\text{Ind}\left((x,k_{x},k_{y})=(-\delta,0,0)\right)=1$.

According to Theorem \ref{thm:Chern}, 
\[
n_{c}\left(E_{1}\rightarrow S_{\delta}^{2}\right)=\text{Ind}\left((x,k_{x},k_{y})=(-\delta,0,0)\right)=1
\]
And from Eq.\,(\ref{eq:ncepsilon}) and Theorem \ref{thm:SpectralFlowIndex},
\[
n_{\text{sf}}=n_{c}\left(E_{1}\rightarrow S_{1}^{2}\right)=n_{c}\left(E_{1}\rightarrow S_{\delta}^{2}\right)=1.
\]
We conclude that there is one net upward spectral flow, i.e., the
TLCW, if the common gap condition is satisfied.

\section{An analytical model for TLCW by a tilted phase space Dirac cone\label{sec:AnalyticalTDC}}

As shown in the Sec.\,\ref{sec:TLCWPrediction}, near the LC Weyl
point only the Langmuir wave and the cyclotron wave are important,
and the $9\times9$ bulk Hamiltonian symbol $H(x,k_{x},k_{y},k_{z})$
can be approximated by the $2\times2$ reduced bulk Hamiltonian symbol
$H_{2}(x,k_{x},k_{y},k_{z})$. For the bulk modes of $H(x,k_{x},k_{y},k_{z})$,
the prominent feature near the LC Weyl point is the tilted Dirac cone
shown in Fig.\,\ref{fig:DiracCone}. This interesting structure is
faithfully captured by $H_{2}(x,k_{x},k_{y},k_{z})$. For comparison,
the tilted phase space Dirac cone of $H_{2}(x,k_{x},k_{y},k_{z})$
is plotted in Fig.\,\ref{fig:DiracConeH2}. From the definition of
$H_{2}(x,k_{x},k_{y},k_{z}),$ it is clear that the factor $4\omega_{\text{pc}}/\alpha^{2}$
is the reason for the cone being tilted.

\begin{figure}[ht]
\centering \includegraphics[width=8cm]{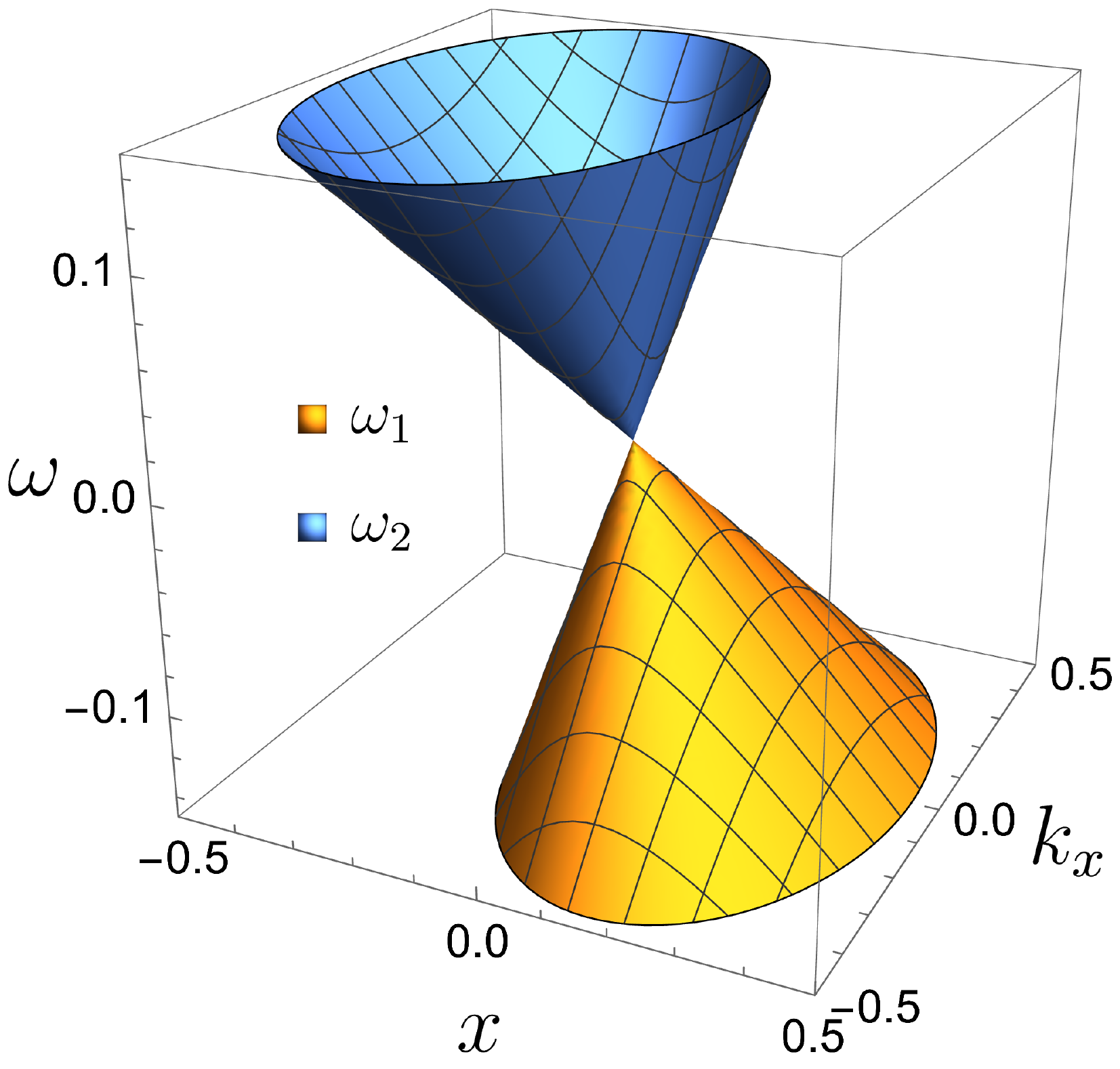} \caption{The tilted phase space Dirac cone of $H_{2}(x,k_{x},k_{y},k_{z})$
at the LC Weyl point. It faithfully represents the tilted Dirac cone
of $H(x,k_{x},k_{y},k_{z})$ shown in Fig.\,\ref{fig:DiracCone}.}
\label{fig:DiracConeH2}
\end{figure}

For $H_{2}(x,k_{x},k_{y},k_{z})$, the corresponding PDO is 
\begin{equation}
\hat{H}_{2}(x,k_{x},k_{y},k_{z})=\begin{pmatrix}\omega_{\text{pc}}-\beta x & \dfrac{\mathrm{i}}{\sqrt{2}\alpha}(\eta\partial_{x}-k_{y})\\[10pt]
\dfrac{\mathrm{i}}{\sqrt{2}\alpha}(\eta\partial_{x}+k_{y}) & \omega_{\text{pc}}+\dfrac{4\omega_{\text{pc}}}{\alpha^{2}}\beta x
\end{pmatrix}.\label{eq:H2hat}
\end{equation}
The theorems proved in Sec.\,\ref{sec:TLCWPrediction} shows that
the full system $\hat{H}(x,-\mathrm{i}\eta\partial_{x},k_{y},k_{z})$
admits one topological edge mode, i.e., the TLCW, as confirmed by
numerical solutions in Secs.\,\ref{sec:Problem-statement} and \ref{sec:Numerial}.
This property of the TLCW is also faithfully captured by the reduced
system $\hat{H}_{2}(x,k_{x},k_{y},k_{z})$.

In particular, Theorem \ref{thm:SpectralFlowIndex} applies to $\hat{H}_{2}(x,k_{x},k_{y},k_{z})$
as well. From Eqs.\,(\ref{eq:om1}) and (\ref{eq:om2}), the common
gap condition is satisfied for $\omega_{1}$ and $\omega_{2}$, and
the proof of Theorem \ref{thm:SpectralFlowIndex} shows that the Chern
number $n_{c}\left(E_{1}\rightarrow S_{1}^{2}\right)$ of the first
eigenmode bundle of $H_{2}(x,k_{x},k_{y},k_{z})$ over $S_{1}^{2}$
equals $1.$ Thus, $\hat{H}_{2}(x,k_{x},k_{y},k_{z})$ admits one
spectral flow, i.e., the TLCW.

To thoroughly understand the physics of a tilted phase space Dirac
cone and the TLCW, we present here the analytical solution of the
entire spectrum of $\hat{H}_{2}(x,k_{x},k_{y},k_{z})$, including
its spectral flow in the band gap. For the PDO corresponding to a
$2\times2$ symbol of a straight Dirac cone, its analytical solution
has been given by Faure \citep{Faure2019}. But for the PDO corresponding
to a $2\times2$ symbol of a tilted Dirac cone, we are not aware of
any previous analytical solution.

To analytically solve for its spectrum, we first simplify the matrix
$\hat{H}_{2}(x,k_{x},k_{y},k_{z})$ in Eq.\,(\ref{eq:H2hat}). Subtract
the entire spectrum by $\omega_{\mathrm{pc}}$ and renormalize $x$
and $k_{y}$ as follows, 
\begin{align}
\tilde{x}:=\sqrt{\dfrac{\sqrt{2}\alpha\beta\kappa}{\eta}}x,\quad\tilde{k}_{y}:=\dfrac{k_{y}}{\sqrt{\sqrt{2}\alpha\beta\eta\kappa}},
\end{align}
where $\kappa^{2}=4\omega_{\mathrm{pc}}/\alpha^{2}$. Matrix $\hat{H}_{2}$
in Eq.\,(\ref{eq:H2hat}) then simplifies to
\begin{align}
\hat{H}_{2}=\sqrt{\dfrac{\sqrt{2}\beta\eta\kappa}{\alpha}}\dfrac{1}{\sqrt{2}}\begin{pmatrix}-\tilde{x}/\kappa & \mathrm{i}(\partial_{\tilde{x}}-\tilde{k}_{y})\\
\mathrm{i}(\partial_{\tilde{x}}+\tilde{k}_{y}) & \kappa\tilde{x}
\end{pmatrix}.
\end{align}
It is clear that $\kappa$ is a parameter measuring how tilted the
Dirac cone is. When $\kappa=1$, the Dirac cone is straight. 

From now on in this section, the overscript tilde in $\tilde{x}$
and $\tilde{k}_{y}$ will be omitted for simple notation. We further
transform $\hat{H}_{2}$ by a similarity transformation and scaling,
\begin{align}
\hat{H}_{2}' & =\sqrt{\dfrac{\alpha}{\sqrt{2}\beta\eta\kappa}}R\hat{H}_{2}R^{-1}=\dfrac{1}{\sqrt{2}}\begin{pmatrix}-x/\kappa & \mathrm{i}(\partial_{x}-k_{y})/\kappa\\
\mathrm{i}\kappa(\partial_{x}+k_{y}) & \kappa x
\end{pmatrix},\\
R & =\mathrm{diag}(\kappa,1).
\end{align}
 $\hat{H}_{2}'$ can be expressed using Pauli matrices and the identity
matrix $\sigma_{0}$ as
\begin{align}
\sqrt{2}\hat{H}_{2}' & =\mathrm{i}(\mu_{2}k_{y}+\mu_{1}\partial_{x})\sigma_{x}+(\mu_{1}k_{y}+\mu_{2}\partial_{x})\sigma_{y}-\mu_{1}x\sigma_{z}+\mu_{2}x\sigma_{0},\\[5pt]
\mu_{1} & =\dfrac{1}{2}\left(\kappa+\dfrac{1}{\kappa}\right),\quad\mu_{2}=\dfrac{1}{2}\left(\kappa-\dfrac{1}{\kappa}\right).
\end{align}

We next apply a unitary transformation to cyclically rotate Pauli
matrices such that $(\sigma_{x},\sigma_{y},\sigma_{z},\sigma_{0})\to(\sigma_{y},\sigma_{z},\sigma_{x},\sigma_{0})$.
Under this rotation, $\hat{H}_{2}'$ becomes 
\begin{align}
\hat{H}_{2}''=\begin{pmatrix}\mu_{1}\lambda+\mu_{2}\hat{a} & \mu_{2}\lambda-\mu_{1}\hat{a}^{\dagger}\\
-\mu_{2}\lambda-\mu_{1}\hat{a} & -\mu_{1}\lambda+\mu_{2}\hat{a}^{\dagger}
\end{pmatrix},\label{eq:H2transform}
\end{align}
where $\lambda=k_{y}/\sqrt{2}$ and 
\begin{align}
\hat{a}=\dfrac{1}{\sqrt{2}}(x+\partial_{x}),\quad\hat{a}^{\dagger}=\dfrac{1}{\sqrt{2}}(x-\partial_{x}).
\end{align}
are annihilation and creation operators. Notice that $\mu_{1}=1$
and $\mu_{2}=0$ when $\kappa=1$, and this is the special case when
$\hat{H}_{2}''$ reduces to a Hamiltonian corresponding to a straight
Dirac cone \citep{Faure2019,Delplace2022}. We now construct an analytical
solution of $\hat{H}_{2}''$. 

Recall that the eigenstates of a quantum harmonic oscillator $|n\rangle$
can be represented by the Hermite polynomials $H_{n}(x)$ as
\begin{align}
\langle x|n\rangle=\varphi_{n}(x)=\frac{1}{\left(2^{n}n!\sqrt{\pi}\right)^{1/2}}\mathrm{e}^{-\frac{x^{2}}{2}}H_{n}(x).
\end{align}
Define a set of shifted wave functions $|n;\delta\rangle$ by 
\begin{equation}
\langle x|n;\delta\rangle:=\varphi_{n}(x+\sqrt{2}\delta).
\end{equation}
They satisfy the following iteration relations,
\begin{align}
\hat{a}^{\dagger}|n;\delta\rangle & =\sqrt{n+1}|n+1;\delta\rangle-\delta|n;\delta\rangle,\\
\hat{a}|n;\delta\rangle & =\sqrt{n}|n-1;\delta\rangle-\delta|n;\delta\rangle.
\end{align}

With these shifted wave functions as basis, it can be verified that
$\hat{H}_{2}''$ has two sets of eigenvectors,
\begin{align}
\psi_{n}^{\pm}=\begin{pmatrix}|n+1;\delta_{n}^{\pm}\rangle\\
\gamma_{n}^{\pm}|n;\delta_{n}^{\pm}\rangle
\end{pmatrix},\quad & n=0,1,2,\cdots\thinspace,\label{eq:eigenmods}
\end{align}
where
\[
\gamma_{n}^{\pm}=\dfrac{\sqrt{n+1}}{-\lambda\mp\sqrt{\lambda^{2}+n+1}},\quad\delta_{n}^{\pm}=\pm\dfrac{\mu_{2}}{\mu_{1}}\sqrt{\lambda^{2}+n+1}.
\]
The corresponding eigenvalues are 
\begin{align}
E_{n}^{\pm}=\pm\dfrac{2\kappa}{1+\kappa^{2}}\sqrt{\lambda^{2}+n+1},\quad n=0,1,2,\cdots\thinspace.
\end{align}

Importantly, there is one additional eigenstate that is not included
in Eq.\,(\ref{eq:eigenmods}), which, in fact, represents the spectrum
flow. Its eigenvector and eigenvalue are 
\begin{align}
\psi_{-1}=\begin{pmatrix}|0;\delta_{-1}\rangle\\
0
\end{pmatrix},\quad E_{-1}=\dfrac{2\kappa}{1+\kappa^{2}}\lambda,
\end{align}
where
\begin{align}
\delta_{-1}=\dfrac{\mu_{2}}{\mu_{1}}\lambda.
\end{align}
Here, we abusively denote this eigenmode as the ``$n=-1$'' eigenstate.
The spectral flow is a linear function of $k_{y}$, and its mode structure
is a shifted Gaussian function. 

The spectrum of $\hat{H}_{2}''(x,-\mathrm{i}\partial_{x},\lambda)$
are plotted in Fig.\,\ref{fig:TDCspectrum}(a). The spectrum consists
of three parts, the upper and lower parts are the global modes in
the frequency bands of $H_{2}(x,k_{x},k_{y},k_{z})$. The middle part
is the single spectral flow connecting the left of the lower part
to the right of the upper part. Note that the tilted Dirac cone of
$H_{2}(x,k_{x},k_{y},k_{z})$ breaks up into two pieces in the global
modes of $\hat{H}_{2}''(x,-\mathrm{i}\partial_{x},\lambda)$. In Fig.\,\ref{fig:TDCspectrum}(b),
the analytical solution of the mode structure of the spectral flow
of $\hat{H}_{2}''(x,-\mathrm{i}\partial_{x},\lambda)$ is plotted.
The analytical result displayed in Fig.\,\ref{fig:TDCspectrum} agrees
well with the numerical solution shown in Fig.\,\ref{fig:1Dspectrum}.

\begin{figure}[ht]
\centering \includegraphics[height=5.7cm]{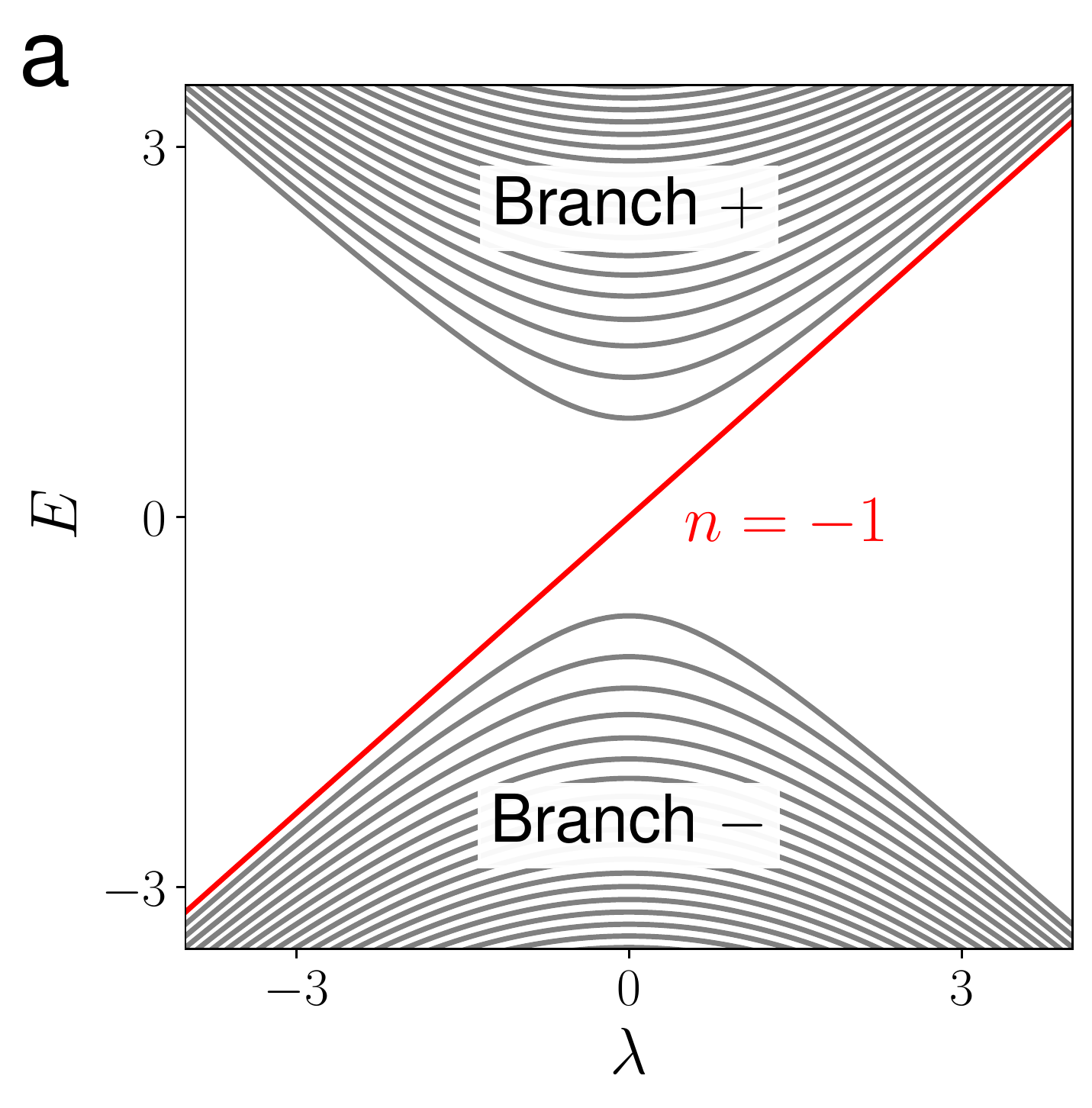} \hspace{0.5cm}
\includegraphics[height=6cm]{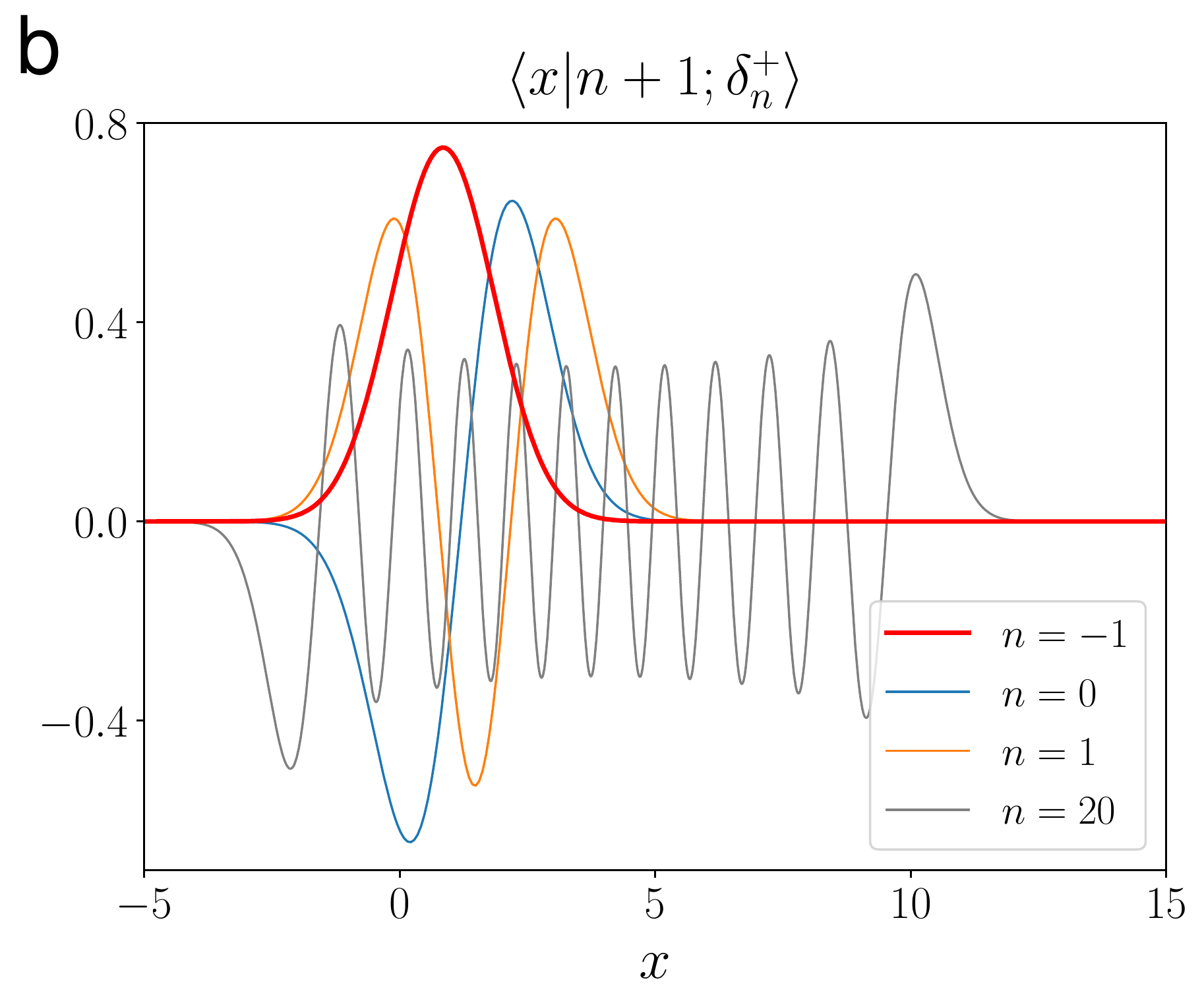} \caption{(a) Analytical spectrum of $\hat{H}_{2}''(x,-\mathrm{i}\partial_{x},\lambda)$
as a function of $k_{y}$. (b) Analytical mode structure of the TLCW.
The result agrees well with the numerical solution shown in Fig.\,\ref{fig:1Dspectrum}.}
\label{fig:TDCspectrum}
\end{figure}

\section{Conclusions and discussions \label{sec:Conclusions} }

Inspired by advances in topological materials in condensed matter
physics \citep{thouless1982quantized,halperin1982quantized,hasan2010colloquium,bernevig2013topological,qi2011topological,armitage2018weyl},
study of topological waves in classical continuous media, such as
electromagnetic materials \citep{silveirinha2015chern,silveirinha2016bulk,gangaraj2017berry,marciani2020chiral},
fluid systems \citep{delplace2017topological,Faure2019,perrot2019topological,tauber2019bulk,venaille2021wave,zhu2021topology,souslov2019topological,qin2019kelvin,fu2020physics,David2022},
and magnetized plasmas \citep{gao2016photonic,yang2016one,parker2020nontrivial,parker2020topological,parker2021topological,fu2021topological,Fu2022a,Rajawat2022,qin2021spontaneous},
has attracted much attention recently. The Topological Langmuir-Cyclotron
Wave (TLCW) is a recently identified topological surface excitation
in magnetized plasmas generated by the nontrivial topology at the
Weyl point due to the Langmuir wave-cyclotron wave resonance \citep{fu2021topological,Fu2022a}.
In this paper, we have systematically developed a theoretical framework
to describe the TLCW.

It has been realized that the theoretical methodology for studying
topological material properties in condensed matter physics cannot
be directly applied to classical continuous media, because the momentum
(wavenumber) space for condensed matter is periodic, whereas that
for classical continuous media is not. Specifically, the typical momentum
space for classical continuous media is $\mathbb{R}^{n}$ $(n=1,2,3),$
and it is difficult to integrate the Berry curvature over $\mathbb{R}^{n}$
to obtain an integer number that can be called the Chern number. The
difficulty has been attributed to the fact that $\mathbb{R}^{n}$
is not compact, and different remedies have been proposed accordingly.
However, we demonstrated that the key issue is not whether the momentum
space is non-compact, but rather that it is contractible. When the
base manifold is contractible, all vector bundles on it are topologically
trivial, and whether an integer index can be designed is irrelevant.
For classical continuous media, nontrivial topology can be found only
for vector bundles over phase space. Without modification, the Atiyah-Patodi-Singer
(APS) index theorem \citep{Atiyah1976} proved for spectral flows
over $S^{1}$ is only applicable to condensed matters, and Faure's
index theorem \citep{Faure2019} for spectral flows over $\mathbb{R}$-valued
$k_{y}$ should be adopted for classical continuous media.

In the present study, the TLCW is defined as a spectral flow of a
Pseudo-Differential-Operator (PDO) $\hat{H}$ for plasma waves in
an inhomogeneous magnetized plasma, and the semi-classical parameter
of the Weyl quantization operator is identified as the ratio between
electron gyro-radius and the scale length of the inhomogeneity. We
formally constructed the Hermitian eigenmode bundles of the bulk Hamiltonian
symbol $H$ corresponding to the PDO $\hat{H},$ and emphasized that
the properties of spectral flows are determined by the topology of
the eigenmode bundles over non-contractible phase space manifolds.
To calculate Chern numbers of eigenmode bundles over a 2D sphere in
phase space, as required by Faure's index theorem, a boundary isomorphism
theorem (Theorem \ref{thm:BoundaryIso}) was established.

The TLCW is proved to exist in magnetized plasmas as a spectral flow
with the spectral index being one. The Chern theorem (Theorem \ref{thm:Chern}),
instead of the Berry connection or any other connection, was used
to calculate the Chern numbers. Finally, we developed an analytically
solvable model for the TLCW using a tilted phase space Dirac cone.
An analytical solution of the PDO of a generic tilted phase space
Dirac cone was found, which generalized the previous result for a
straight Dirac cone \citep{Faure2019}. The spectral flow index of
the tilted Dirac cone was calculated to be one, and the mode structure
of the spectral flow was found to be a shifted Gaussian function.

As a topological edge wave, the TLCW can propagate unidirectionally
and without reflection and scattering along complex boundaries. Due
to this topological robustness, it might be relatively easy to excite
the TLCW experimentally. Of course, laboratory and astrophysical plasmas
are subject to many more physical effects that have not been included
in the present model, such as collisions and finite temperature. For
practical application, these factors need to be carefully evaluated
by experimental and theoretical methods. 
\begin{acknowledgments}
This research was supported by the U.S. Department of Energy (DE-AC02-09CH11466).
We thank Dr. F. Faure, Dr. P. Delplace, and Prof. B. Simon for fruitful
discussion. The present study is inspired by their groundbreaking
contributions. 
\end{acknowledgments}

\bibliography{ref}

\end{document}